\documentclass[%
 reprint,
nofootinbib,
 amsmath,amssymb,
 aps,
]{revtex4-1}

\usepackage{graphicx}% Include figure files
\usepackage{dcolumn}% Align table columns on decimal point
\usepackage{bm}% bold math
\usepackage{braket}
\usepackage{graphicx}% Include figure files
\usepackage{hyperref}% add hypertext capabilities
\usepackage[utf8]{inputenc}
\usepackage{amsthm}
\usepackage{bbold}
\usepackage{tikz}
\usepackage{amsmath}
\usepackage{caption}
\usepackage{subcaption}
\usepackage{tkz-euclide}
\usepackage{adjustbox}
\usepackage{tikz}
\usepackage{pst-solides3d}
\usetikzlibrary{matrix}
\usepackage{tikz-3dplot}
\usepackage[colorinlistoftodos,prependcaption,textsize=small]{todonotes}
\usepackage[shortlabels,inline]{enumitem}
\theoremstyle{plain}% default
\newtheorem{thm}{Theorem}[section]

\newtheorem{prop}[thm]{Proposition}
\theoremstyle{definition}
\newtheorem{defn}[thm]{Definition}

\theoremstyle{remark}
\newtheorem{rem}[thm]{Remark}

\begin{document}

\preprint{APS/123-QED}

\title{%A 
Dark Energy from %of 
topology%ical origin% 2.0
}% Force line breaks with \\
%\thanks{We should change this title, I think it is misleading}%

\author{J. Lorca Espiro$^1$}
% \altaffiliation[Also at ]{}%Lines break automatically or can be forced with \\
%\author{}%
% \email{Second.Author@institution.edu}%
\email{javier.lorca@ufrontera.cl}
\affiliation{$1$ - Departamento de Ciencias F\'\i sicas, Facultad de Ingenier\'\i a, Ciencias y Administraci\'on, Universidad de La Frontera, Avda. Francisco Salazar 01145, Casilla 54-D Temuco, Chile.}
%\collaboration{MUSO Collaboration}%\noaffiliation

\author{M. Le Delliou$^{2,3,4}$}
 %\homepage{http://www.Second.institution.edu/~Charlie.Author}%
\email{{\it Corresponding author:}\\(delliou@lzu.edu.cn,) morgan.ledelliou.ift@gmail.com}
\affiliation{
$2$ - Institute of Theoretical Physics, School of Physical Science and Technology, Lanzhou University,
No.222, South Tianshui Road, Lanzhou, Gansu 730000, P R China
}%
\affiliation{
$3$ - Instituto de Astrof\'isica e Ci\^encias do Espa\c co, Universidade de Lisboa,
Faculdade de Ci\^encias, Ed.~C8, Campo Grande, 1769-016 Lisboa, Portugal
}%
\affiliation{$4$ - Instituto de F\'isica Te\'orica, Universidade Estadual de S\~ao Paulo (IFT-UNESP), Rua Dr. Bento Teobaldo Ferraz 271, Bloco 2 - Barra Funda,  CEP 01140-070, S\~ao Paulo, SP, Brazil}%

%\author{M. Fontanini$^5$}
%\affiliation{%
% $5$ - Universit\'e Paris Diderot-Paris 7, APC-Astroparticule et Cosmologie (UMR-CNRS 7164), 
%Batiment Condorcet, 10 rue Alice Domon et L\'eonie Duquet, F-75205 Paris Cedex 13, France.}%

%\collaboration{CLEO Collaboration}%\noaffiliation

\date{\today}% It is always \today, today,
             %  but any date may be explicitly specified

\begin{abstract}
The concordance model of cosmology suffers from the major theoretical problems surrounding the observed value and recent emergence of a cosmological constant. In this paper we present a novel approach, which explains more naturally its value than that based on quantum vacuum energy, in the form of topological invariants characteristic classes, included as Lagrange multipliers in the action. The approach draws from topological as well as dynamical system consideration, generating as a byproduct an effective cosmological constant. General Relativity is recovered by canceling the torsion in a region containing the observable Universe, which boundary constraints the invariants, thus yielding the effective cosmological constant's form. As that form's denominator contains the total volume of the average black hole, calculated from a geometrical mean on the estimated black hole mass distribution and directly associated to the ratio of the total volume boundary of the space-time manifold and the dominant term in its Euler characteristic. The constant's small estimated value compared to the Planck scale is therefore natural and our evaluation fits remarkably well with the observed value.
%\begin{description}
%\item[Usage]
%Secondary publications and information retrieval purposes.
%\item[PACS numbers]
%May be entered using the \verb+\pacs{#1}+ command.
%\item[Structure]
%You may use the \texttt{description} environment to structure your abstract;
%use the optional argument of the \verb+\item+ command to give the category of each item. 
%\end{description}
\end{abstract}

%\pacs{Valid PACS appear here}% PACS, the Physics and Astronomy
                             % Classification Scheme.
%Use showkeys class option if keyword
                              %display desired
\maketitle

\tableofcontents

\section{\label{sec:intro}Introduction}
The need for a cosmological constant, or an equivalent effect, was revived when stronger than expected fainting in type Ia supernovae, used as standard candles, was detected \cite{Perlmutter:1997zf,Riess:1998cb,Perlmutter:1998np}. This was interpreted as cosmic acceleration and confirmed later from cosmic microwave background radiation, clusters and baryon acoustic oscillation measurements \cite{Hinshaw:2012aka,Bennett:2012zja,Aghanim:2012vda,Ade:2013sjv}.

The first common interpretation of such constant portraits it as quantum vacuum energy. However this raises more problems, such as that from naive evaluations of the full quantum vacuum energy, i.e. of order the Planck scale \cite{Martin}, compared with the $\Lambda$ value extracted from observed acceleration, offers one of the largest discrepancy in physics, coined the fine tuning problem \citep{Weinberg:1988cp,Weinberg:2008zzc}.

The fine tuning problem proceeds from the dominance at early time of the Planck scale on the quantum vacuum scale. It should have therefore given initial values for all running constants, including the cosmological constant. Since $\Lambda$ is constant, its value is naively expected to be set from the energy conditions of the Planck era, instead of the hundred order of magnitude smaller observed value from cosmic dynamics. Moreover, as the cosmological constant is non-varying, it also entails the coincidence problem \citep{Amendola:1999er,TocchiniValentini:2001ty,Zimdahl:2001ar,Zimdahl:2002zb}. 

That extra problem stems from the apparently unnatural proximity between the present epoch and the moment of emergence of the cosmological constant as dominant in the cosmic energy density balance \cite{Frieman:2008sn,Li:2012dt}. Such surprising coincidence between values of $\Lambda$, expected to come from the Universe evolution's initial conditions, and today's Universe matter energy density content  would require at least an explanation.

The standard understanding of the cosmological constant is review, e.g. in \cite{Martin}.

General relativity's (GR) curved spacetime can be generalised, to include torsion on the manifold, in Einstein-Cartan theories \cite{Dona,Giulini,nakahara,hatcher}. Their two limiting cases for equivalent degrees of freedom allow to recover 
\begin{enumerate}
    \item the torsion-less GR limit and
    \item the curvature-less Teleparallel Equivalent to GR case \cite{Deandrade,Baez,Arcos}.
\end{enumerate} 
This work focusses on the GR limit of Einstein-Cartan theories.

The Gravitational action contains further freedom in the form of boundary terms \cite{Baekler,Dyer} and manifold topology \cite{donaldson,Sengupta,Kaul}. We restrain here the latter with an action, built on the Einstein-Hilbert (EH) and Vielbein-Einstein-Palatini (VEP) equivalence \cite{Dadhich}, that adds to GR's well known features, the spacetime topology-consistent Chern-type topological invariants, as well as the cobordism.

Such approach remarkably gives rise to a topologically related term which, under certain conditions, acts as an \textit{effective cosmological constant}.

This very general action allows this paper to argue that the cosmological constant can have a topological origin rather than being generated at the universe's formation from quantum vacuum considerations. The restriction of such action, naturally featuring torsion and an %a topologically related 
\textit{effective cosmological constant} related to topology, to the usual GR condition can be effected by dynamical systems stability considerations. Since the resulting \textit{effective cosmological constant} connects to the topology of the Universe via a value we argue to be proportional to the finite inverse average black hole volume, and no longer to the large Planck density, it renders moot the fine tuning problem. A preliminary discussion of these results can be found in \cite{LeDelliou:2019mus}. This paper presents a detailed technical account of this work.

Topology as a source of cosmology has been considered in the works of \cite{Alexander1} and \cite{Alexander2}.In those papers, the topological terms are coupled trivially over closed manifolds, contrary to our approach. This allows for the classical field equations to remain unchanged, only affecting the quantum equations. This is obtained using the Ashtekar variables formalism, while we keep our work at the classical level and use Einstein-Cartan formalism.

We begin by introducing the topological invariants articulation with the Vielbein-Einstein-Palatini action in Sec.~\ref{sec:model}. We introduce the supplementary dynamics induced by these invariants in the action, which result, in Sec. \ref{sec:cosmtop}, in the effective cosmological constant and, as a byproduct and on a practical ground, in the recovery of GR. Section \ref{sec:TCCeval} presents how this topological constant value can be estimated for our Universe. Finally, we discuss and conclude our findings in Sec.~\ref{sec:Conclusion}.

\section{\label{sec:model}The model}

In this section we briefly outline the mechanism in which the torsion-less classical gravitational action inherits torsional degrees of freedom when topological terms are added. By topological terms we mean the so called characteristic classes well defined for any $4n$-manifold, in this $n=1$. This procedure will suggest the natural upgrade of the cosmological constant term in the \textit{Einsteinian} tradition of gravitational theories to a functional with degrees of freedom related to topological information of the space-time manifold. Appendix \ref{sec:cartan} provides a brief account of the formalism and overall notation.

In \ref{subsec:comact} is demonstrated that, in the case in which the space-time manifold $\mathcal{M}$ has a finite volume $\text{Vol} \left( \mathcal{M} \right)$ and a boundary $\partial \mathcal{M}$, the cosmological constant $\Lambda$ is functionally related to the Dirichlet to Neumann operator $\Upsilon \left[ \chi \right]$ (DN operator). Recall that this operator maps the Dirichlet boundary conditions to those of Neumann for any well posed boundary value problem\footnote{Informally, a physical observable $\mathcal{O}$, associated to differential form $\hat{\mathcal{O}} \in \Omega \left( \mathcal{M} \right)$, can be uniquely defined as a Dirichlet problem, i.e. by knowing its value over space-time 
boundary $\hat{\mathcal{O}} |_{ \partial \mathcal{M}}$, or as a Neumann problem, i.e. by knowing the value of its normal derivatives $L_n \hat{\mathcal{O}} |_{ \partial \mathcal{M}}$ over the space-time boundary, with $n$ the outward unit vector over the boundary $\partial \mathcal{M}$. The Dirichlet to Neumann operator allows to move freely between the two types of problems.}. The exact expression is simply:
\begin{align*}
\Lambda = \kappa \Upsilon \left[ \chi \right] = \frac{1}{ \text{Vol} \left( \mathcal{M} \right)} \int\limits_{\partial \mathcal{M}} \chi \;\;,
\end{align*}
which is valid for $\chi$, a $3$-form satisfying the well posed boundary value problem appearing in (\ref{DNop}) and Proposition \ref{DNoperator}, with $\kappa = 8 \pi G c^{-4}$ being the gravitation coupling constant. This result strikes as being of a mostly topological nature. Moreover, if the kernel of the DN operator $\Upsilon \left[ \chi \right]$ is known, lower bounds for the so called Betti numbers (see Eq.~\ref{betti}) can be obtained, which we can intuitively understand as counting the $n$-dimensional punctures of a manifold. Hence, the topology of the manifold $\mathcal{M}$ is restricted.

This observation enlighten us is two ways: 

\begin{enumerate}[i)]
\item The cosmological constant, for the conditions presented, is related to topological rather than vacuum energy considerations. Furthermore, it is inversely proportional to $\text{Vol} \left( \mathcal{M} \right)$ which naturally favors small values for $\Lambda$, opening the possibility to alleviate or even solve the fine tuning problem.

\item  The cosmological constant can be understood as some type of functional of certain fields rather than as a fixed parameter of the theory. 
\end{enumerate}

Hence, inspired by the above observations we upgrade the cosmological constant $\Lambda$ to a functional $\tilde{\lambda}$. We expect this new quantity to also be characterized by topological degrees of freedom such that this new functional defines a well posed variational problem when the field equations are tackled.

The model is completely characterized in Appendix \ref{sec:The model} but here are discussed its most important features. Let the index set $J = \left\{ N, P , E \right\} $ with $j \in J$ be the collection of subscripts related to the Nieh-Yan, Pontryagin and Euler characteristic classes, respectively. We focus on the family of actions having the following three distinctive sectors: 
\begin{align}\label{Bodyaction}
S & := S_{G} \left[e^c, \omega^a_{\;\;b} , \tilde{\lambda} \right] + S_{T} \left[e^c, \omega^a_{\;\;b}, \varphi_j \right] +  S_M \left[ \Psi \right] \; ,
\end{align}
where the first term is the pure gravitational action (see Eq.~\ref{EHactionmod}), $e^c$ are the tetrad, $\omega^a_{\;\;b}$ is the total connection and $\tilde{\lambda}$ is the newly upgraded \textit{cosmological functional}. The second term is the \textit{topological action}, named after the fact that it explicitly includes the three characteristic classes $C_j$ $(j \in J)$ for the characteristic classes: 
\begin{align}
\label{Pontryagin1} \small C_{P}  & = \frac{1}{8 \pi^2} R^a_{\;\;b} \wedge R^b_{\;\;a} \;\;,\\
\label{Euler1} \small C_{E} & = \frac{1}{8 \pi^2} R^{ab} \wedge R^{\left( * \right)}_{ba} \;\;, \\
\label{NiehYan1}\small C_N & = T^a \wedge T_a - R_{ab} \wedge \Sigma^{ab} \;\;,
\end{align}
(See also \ref{TopManifold}) which are consistent with the space-time manifold $\mathcal{M}$. Each characteristic class is coupled to its corresponding $\varphi_j$ zero-forms (see Eq.~ \ref{Stop}) (functions) having dimensions of $ \left[ \texttt{length}^{2} \right]$ for $j=E,P$ and being dimension-less for $j=N$.

The last term in (\ref{Bodyaction}) stands for an explicit matter action (see Eq.~\ref{SMat}) of the Dirac type, where $\Psi:=\left\{ \bar{\psi}, \psi \right\}$ is a collected term for the spinor fields. As usual, a matrix representation $\gamma^a$ of the Clifford algebra $\left\{ \gamma^a , \gamma^b \right\} = 2 \eta^{ab} $ ( with $\eta^{ab}$ the Minkowskian metric) is considered, such that $\bar{\psi }:= \psi ^{\dagger} \gamma^{0}$ and $\psi^{\dagger}$ is the conjugate transpose of $\psi$. By specifying the cosmological functional to behave as the Yukawa-type of interaction
\begin{align}\label{Yukawa}
\tilde{\lambda} := \bar{\psi} \lambda \left[ \varphi_j \right] \psi \;\; ,    
\end{align}
we are in position to obtain the variation of Eq.~(\ref{Bodyaction}). Notice that $\lambda$ now concentrates all the information and becomes the only source of interaction between the spinor fields and the coupling zero-forms. The field equations (for $\delta e^a$, $\delta \omega_{ab}$, $\delta \varphi_j$, $\delta \bar{\psi}$ and $\delta \psi$ , respectively)\footnote{Originally the field equations suppose the presence of a fermion potential $V \left( \Psi \right)$. However, since this potential is not involved in the final calculation, we prefer to omit it out of clarity.} read:
\begin{align}
\label{fe1} & 0 = \left[ R^{\left( *\right)}_{ab} - \frac{2 \lambda}{3} \left| \psi \right|^2_{\psi} \hat{\Sigma}^{(*)}_{ab} \right] \wedge e^{b} + i d \varphi_{N} \wedge T_{a} + \Xi_a  \;\; , \\
\nonumber & 0 = d_{\omega} \hat{\Sigma}^{(*)}_{ab} + i 2 d \hat{\varphi}_P \wedge R_{ab} - i 2 d \hat{\varphi}_E \wedge R^{\left( * \right)}_{ab}  + \\ 
\label{fe2} & \quad \quad \quad \quad \quad - i d \varphi_{N} \wedge \hat{\Sigma}^{\flat}_{ab} - \tau_{ab} \;, \\ 
\label{fe3} & 0  = - i \frac{\delta \lambda}{ \delta \varphi_j } \left| \psi \right|^2_{\psi} d \mu - C_j \;\;\;\;\;\; , \;\;\;\;\;\; j = \left\{ E, P, N \right\} \;\; , \\
\label{fe4} & \; 0 = \gamma^a D_{\omega} \psi - \frac{\lambda}{4} \psi \, e^a \;\;\;\;\; , \;\;\;\; 0 = D_{\omega} \bar{\psi} \gamma^a + \frac{\lambda}{4} \bar{\psi} \, e^a \;\;\; ,
\end{align}
where we have defined $\left| \psi \right|^2_{\psi} := \bar{\psi} \psi$ and the normalized quantity $\hat{\varphi}_j := \frac{\varphi_j}{\left( 4 \pi \right)^2}$ for brevity, $d_{\omega}$ stands for the exterior covariant derivative, $D_{\omega}$ stands for the gauge covariant derivative (see definition in Eq.~\ref{gaugecovder}) and we have defined the quantities:
\begin{align}
\label{fevarious} & \Xi_{a} := \frac{ \delta \mathfrak{L}_{M}}{\delta e^a} = 2 \mathfrak{Re} \left\{ \bar{\psi} \gamma^{b} D_{\omega} \psi \right\}\wedge \hat{\Sigma}^{(*)}_{ab} \;\;
; \;\; i := \sqrt{-1} \;\;\; ; \;\;\;  \\
\nonumber & \tau_{ab} := \frac{\delta \mathfrak{L}_M}{\delta \omega^{ab} } = \frac{1}{4} \bar{\psi} \sigma_{abc} \psi \star e^c \; ; \; \sigma^{a_1 \cdots a_k}:= \gamma^{\left[ a_1 \right.} \cdots \gamma^{\left. a_k \right]} \, .
\end{align}

Notice the non-GR terms in the field Eqs.~(\ref{fe1}) and (\ref{fe2}), i.e. the ones that are multiplied by the coupling zero-forms and/or its exterior derivatives. These are the terms that ultimately differentiate the current theory with that of the strict GR.

Being so, we can interpret the cosmological functional $\tilde{\lambda}$ as the source of the topology of the space-time manifold $\mathcal{M}$. As discussed in Appendix \ref{TopManifold}, a direct consequence of this modelling is that it ensures the space-time manifold $\mathcal{M}$ to be oriented, simply connected and smooth.

In this way, we show that the couplings zero forms $\varphi_j$ ($j \in J$) can actually be naturally interpreted as scalar matter fields. This becomes clear once the total $2$-form curvature $R_{ab}$ (see Eq.~\ref{Riemanntensorgral}) is obtained, its torsion related part will have terms directly dependant on $\varphi_j$ with a clear dynamical interpretation.

\section{\label{sec:cosmtop}Obtaining a cosmological constant from topology}

A thorough exposition of the detailed calculations is presented in Appendix \ref{sec:sol}. We outline here the general procedure used to solve the field equations. We stress that, given the strong non-linearities of the field equations, we opted for an approach mixing formal and heuristic methodologies in which some physical restrictions have been made in order to obtain closed analytical solutions.

\subsection{\label{subsec:nonlinear} Managing the non-linearities of the canonical field equations}

Basically, due to the form of the topological action, we are allowing for the presence of torsion in the field equations, so the VEP first order formalism becomes ideal to treat this feature. We recall that the total connection $1$-form $\omega^{a}_{\;\;b}$ can always be written as:
\begin{align*}
\omega^{a}_{\;\;b} & = \bar \omega^{a}_{\;\;b} + K^{a}_{\;\;b}  & \in \Omega^1 \left( \mathcal{M} \right) \; ,
\end{align*}
where $\bar \omega^{a}_{\;\;b}$ is the torsion-less Levi-Civita connection, while $ K^{a}_{\;\;b}$ is the torsion related contortion $1$-form. The latter becomes the only responsible for the torsion $2$-form as: 
\begin{align*}
T^a & = K^{a}_{\;\;b} \wedge e^b & \in \Omega^2 \left( \mathcal{M} \right) \; .
\end{align*}

By means of the same connection decomposition, the total curvature $2$-form can then be written as (see Appendix \ref{subsec:VEPformalism}, Eq.~(\ref{torcurv}) for details):
\begin{align}
R^a_{\;\;b} & =  \bar{R}^a_{\;\;b} +  \Theta^a_{\;\;b}  & \in \Omega^2 \left( \mathcal{M} \right) \;\;,\label{RiemanntensorgralB}
\end{align}
where $\bar{R}^a_{\;\;b}$ is the torsion-less part of the curvature, formally equivalent to that of a GR curvature $2$-form, while the term $\Theta^a_{\;\;b}$ concentrates all the contributions from the torsion related quantities by means of the contortion $K^{a}_{\;\;b}$ as:
\begin{align}\label{torcurvB}
\Theta^a_{\;\;b} := d_{\omega} K^a_{\;\;b} - K^a_{\;\;c} \wedge K^{c}_{\;\;b} \quad \in \Omega^2 \left( \mathcal{M} \right) \;\;.
\end{align}

Focusing back on the first two field Eqs.~(\ref{fe1}) and (\ref{fe2}), it can be seen that a GR-like behavior is obtained back if the coupling zero-forms $\varphi_j$ ($j \in J$) become slowly varying functions within the torsion-less regions of space-time. The possibility of this behavior for the coupling zero-forms is related to the existence of stability regions (see Appendix \ref{sec:borel} and the references therein) for a turn around in the dynamical system defined by the field Eqs.~(\ref{fe1} - \ref{fe4}). We can therefore expect to relate the vanishing of the torsion with the derivatives of such a slow varying function. This can be equivalently thought of as a contortion which is at least proportional to the derivatives of the coupling zero-forms $\varphi_j$ ($j \in J$). Consequently, this observation plus the natural decomposition of the curvature $2$-form $R_{ab}$ (\ref{RiemanntensorgralB}), suggests a contortion $1$-form:
\begin{align}\label{contortionbody}
K_{ab} = 4i \star \left( d \varphi_N \wedge \hat{\Sigma}^{\flat}_{ab} \right) \;\; \in \;\; \Omega^1 \left( \mathcal{M} \right) \;\;\;.
\end{align}
This contortion cancels the undesired non-linear terms appearing in the field Eqs.~(\ref{fe1}) and (\ref{fe2}) via the torsion related part of the curvature $\Theta_{ab}$ (See Eq.~(\ref{torcurvB})).

It is through $\Theta_{ab}$ also that a term proportional to the inner product $\left| d \varphi_N \right|^2$ appears in the field equations. When interpreted as a kinetic term, the coupling zero-form $\varphi_N$ inherits a matter field status. It is then just a matter of rearranging the action  (\ref{Bodyaction}), assuming the correctness of the contortion (\ref{contortionbody}), so that the coupling zero-forms can be gathered in a canonical scalar matter field action, justifying our claim on the interpretation of them as matter fields. Moreover, after some algebra and with the help of Eqs.~(\ref{fe4}) to (\ref{fe1}), we obtain the following expression for the torsion-less curvature $2$-form:
\begin{align*}
& \bar{R}_{ab} = \frac{ 2 \left( \alpha_{\psi} \lambda \left[ \varphi_j \right] + \left| d \varphi_N \right|^2 \right)}{3} \hat{\Sigma}^{\flat}_{ab} \;\; ; \;\; \alpha_{\psi} := \frac{\left| \psi \right|^2_{\psi}}{4} \in \mathbb{R} \;\; .
\end{align*}
which is formally equivalent to the analogous field equation obtained in a GR-like theory. The last expression has the right behavior in the limiting case where torsion is turned off, as expected.

Tackling the second field equation, we see that (\ref{fe1}) only contains one of the coupling zero form, i.e. $\varphi_N$, while in comparison (\ref{fe2}) had all three $\varphi_j$ ($j \in J$). Thus, in order for the expression $\Theta_{ab}$, depending only on the coupling zero-form $\varphi_N$, to be sufficient to correct the non-GR terms in (\ref{fe2}) we should expect some dependency between the coupling zero-forms. Moreover, when inserting the calculated total curvature $R_{ab}$ and its Lie dual $R^{(*)}_{ab}$ (See expression \ref{curvaturetotal}) into the second field equation (\ref{fe2}), we obtain:
\begin{align*}
& \tau_{ab} = d_{\omega} \hat{\Sigma}^{(*)}_{ab} - \frac{4i}{3} \left[ \alpha_{\psi} \lambda d \hat{\varphi}_E + \left( i 2 \Delta \varphi_N \right) d \hat{\varphi}_P \right] \wedge \hat{\Sigma}^{(*)}_{ab} + \\ 
& + \frac{4 i}{3} \left[ \alpha_{\psi} \lambda d \hat{\varphi}_P - \left( i 2 \Delta \varphi_N \right) d \hat{\varphi}_E - \frac{3}{4} d \varphi_{N} \right] \wedge \hat{\Sigma}^{\flat}_{ab} + \\
& + 2 d \varphi_N  \wedge \left[ \left( d \hat{\varphi}_E \wedge K_{ab} \right) + \left( d \hat{\varphi}_P \wedge K_{ab} \right)^{(*)} \right] \;\; ,
\end{align*}
for which we naturally define the quadrature expressions:
\begin{align*}
d \hat{\varphi}_E & = \frac{3}{4} \frac{ i 2 \Delta \varphi_N }{ \left( \alpha_{\psi} \lambda \right)^2 + \left( 2 \Delta \varphi_N \right)^2} d \varphi_N \;\; , \\
d \hat{\varphi}_P & = \frac{3}{4} \frac{ \alpha_{\psi} \lambda }{ \left( \alpha_{\psi} \lambda \right)^2 + \left( 2 \Delta \varphi_N \right)^2 } d \varphi_N \;\; .
\end{align*}
We are effectively writing $\varphi_E$ and $\varphi_P$ in terms of $\varphi_N$. These relations restore the original Eq.~(\ref{fe2}) into the usual generalized spin equation for the analogous GR-like theories in the VEP formalism. The entire process then is summed up in the transformation $\lambda \left[ \varphi_j \right] \rightarrow \lambda \left[ \varphi_N \right]$.

\subsection{\label{subsec:topocalc} Tackling the topological field equations}

At this point, we are in position of calculating the characteristic classes either directly from (\ref{fe3}) or from the set (\ref{Pontryagin} - \ref{NiehYan}). We must ensure that these two approaches are compatible and, at the same time, we will shed some light on the characterization of the functional $\lambda\left[ \varphi_N \right]$.

For the Nieh-Yan characteristic class $C_N$, those two approaches are immediately compatible if the expression
\begin{align}\label{lambda1}
\lambda \left[ \varphi_N \right] = - \frac{\left| d \varphi_N \right|^2}{\left| \psi \right|^2_{\psi}} + \frac{4 \Lambda}{\left| \psi \right|^2_{\psi}}
\end{align}
is satisfied, where the second term is consistent with our assumption that the functional is constant when torsion is turned off. Hence, we can similarly calculate directly the Euler characteristic class $C_E$, also completely consistent but not necessarily illustrative for our purposes, as well as the Pontryagin characteristic class $C_P$, which needs the following restriction for compatibility:
\begin{align}\label{Dalembertian}
3 \left( \alpha_{\psi} \lambda \right)^2 = \left( 2 \Delta \varphi_N \right)^2 \; \Rightarrow \; \Delta \varphi_N = \pm \frac{\sqrt{3}}{2} \alpha_{\psi} \lambda \;\;.
\end{align}
When combined with Eq.~(\ref{lambda1}) characteristic class, Eq.~(\ref{Dalembertian}) yields the second characterization for the functional $\lambda \left[ \varphi_N \right]$ as:
\begin{align}\label{lambda2}
\frac{\delta \ln \left| \lambda \right| }{\delta \varphi_N} = \frac{\sqrt{3}}{4} \;\; \Rightarrow \;\; \lambda = \frac{4 \Lambda}{\left| \psi \right|^2_{\psi}} \exp \left( \pm \frac{\sqrt{3}}{4} \varphi_N \right) \; ,  
\end{align}
where the amplitude of the exponential function has been normalised coherently with (\ref{lambda1}). If we postulate dynamical stability, then the kinetic term associated to the zero-form $\varphi_N$ should behave as a Lyapunov function and we obtain the following expression:
\begin{align*}
\left| \frac{d \varphi_N}{4}  \right|^2 & = \Lambda \left\{ 1 - \exp \left( - \frac{\sqrt{3} }{4} \left| \varphi_N  \right|  \right) \right\} \;\; ,
\end{align*}
which can be used to characterize the GR-like regions of the space-time manifold $\mathcal{M}$, via the condition $d \varphi_N = 0$, as equivalent to $\varphi_N = 0$. We have identified this region as the submanifold $\mathcal{N} \subseteq \mathcal{M}$.

The exact form for the coupling zero-form $\varphi_N$ must be a result of solving the boundary value problem stated in Proposition (\ref{boundaryvalue}), which is out of the scope of this paper.

\subsection{\label{subsec:estimation} The cosmological constant.}

Calculating now the topological numbers, we obtain:

\begin{enumerate}[i)]
\item the Nieh-Yan number $ n_N = 0 $, from Eq.~(\ref{NYdirect}). This might be considered as the model's partial answer to the absence of
global manifestation of a physical quantity that accounts for presence of torsion in the observation. In other words, the model gives no hint of torsion, neither locally, in the GR-like regions, nor globally, since $n_Y$ is exactly null.

\item the Pontryagin number $n_P = 0$, from Eq.~(\ref{CP12}). Physically, this prevents the space-time manifold $\mathcal{M}$ to have an orientation reversing diffeomorphism \cite{donaldson}. However this also prevents the CPT type of symmetry expected from the quantum substructure. At the same time, having a non-zero Pontryagin number means that $\mathcal{M}$ cannot be the boundary of an oriented compact higher dimensional manifold \cite{donaldson}.

\item the Euler number:
\begin{align}\label{nEbody}
 n_E = - \frac{16 \Lambda^2}{3 \left( 4 \pi \right)^2} \left\| u^2_{\hat{\varphi}_N} \right\|^2_{L^2} \; ; \; u_{\hat{\varphi}_N} = \exp \left( - \frac{\sqrt{3}}{4} \left| \varphi_N \right| \right) \; ,
\end{align}
from Eq.~(\ref{CE12}), which is finite by topology since by construction $n_E \in \mathbb{Z}$ (See Theorem \ref{topnumber}). Considering that $n_E$ can be written as (see Sec.~\ref{subsec:characteristicclasses}):
\begin{align}\label{nEbody1}
n_E = - 2 k^2_E b_3 \;\;\; ,
\end{align}
where we have defined
\begin{align}\label{kEB}
k^2_E := 1 - \frac{1 + b}{ b_3}  \lesssim 1 \quad ;
\end{align}
here $2b$ is the number of $2$-\textit{punctures} and $b_3$ is the number of $3$-\textit{punctures} in $\mathcal{M}$, that is the number of $2$- and $3$- dimensional holes in  $\mathcal{M}$. Since we have not observed $2$-dimensional topological defects but black holes, which induce surface holes in space slices, i.e. $3$-dimensional holes in space-time, are pervasive in observation, we claim $b_3$ should dominate for any sensible space-time manifold $\mathcal{M}$, and interpret it to be essentially the number of black holes in the Universe.
\end{enumerate}

We thus obtain from  Eq.~(\ref{nEbody}) the exact result:
\begin{align}\label{lambdaexactbody}
\Lambda^2 = \frac{ 3 \left( 4 \pi \right)^2 k^2_E b_3 }{8 \left\| u^2_{\hat{\varphi}_N} \right\|^2_{L^2}} \quad ,
\end{align}
for which we can give the estimate (See also the discussion in Sec. \ref{subsec:characteristicclasses}):
\begin{align}\label{uestimatebody} 
0 \leq \left\| u^2_{\hat{\varphi}_N} \right\|^2_{L^2} \simeq \frac{\text{Vol} \left( \partial \mathcal{N} \right)}{\mathcal{C}^2} < \infty
\end{align}
where $\mathcal{C}^2$ is the Cheeger or isoperimetric constant \cite{Benson} with dimensions of $\left[ \texttt{length} \right]$. It has been calculated in certain specific cases. Finally, Eq.~(\ref{uestimate}), combined with (\ref{lambdaexactbody}), (\ref{manifoldregions}) and (\ref{eq:Topolambda1}) we have the topological cosmological constant estimate:
\begin{align}\label{eq:Topolambdabody}
 \Lambda \approx  \frac{ 2 \pi \, k_E \, \mathcal{C} }{ \left\langle \frac{2}{3} \text{Vol} \left( \partial \mathcal{M} \right) \right\rangle^{\frac{1}{2}} }  \;\; .
\end{align}

\section{\label{sec:TCCeval} Topological cosmological constant evaluation}

The resulting expression (\ref{eq:Topolambdabody}) for the topological value of the cosmological constant can be then confronted with its observed value.
To evaluate it, we need to consider the three key components:
\begin{enumerate}
    \item From (\ref{nEbody1}), we have that $k_E=\sqrt{-\frac{n_E}{2b_3}}$, the ratio of the Euler number to what we have interpreted as the number of black holes (BHs) in the manifold ($b_3=\mathcal{N}_{BH}$). Since we argue that for any sensible space-time, $b_3$ should dominate the other terms, we take $k_E\approx 1$,
    \item $\mathcal{C}$, the isoperimetric or Cheeger constant,  measures the ratio between the $3$-surface of boundary and the $4$-volume of the bulk of the spacetime considered, which is generally unknown but has been evaluated in some cases \cite{ASENS, Hoffman} to be of order $\sim 10^1$,
    \item $\langle \text{Vol} \left( \partial \mathcal{M} \right) \rangle$, the average volume of the boundary of the space-time. 
\end{enumerate}
We shall evaluate it in what follows.This detailed evaluation is presented in Appendix \ref{sec:BHsurf}.

\subsection{Average boundary volume}

From the assumption on $b_3$, the evaluation of the space-time boundary follows from assuming
\begin{enumerate}
    \item the boundary is made of BHs,
    \item the hypersurface of the boundary is the sum of all their outer horizons. These are approximated by their Schwarzschild horizons, neglecting Kerr horizons deformations and merger history variations,
    \item the distribution of BH in the Universe is well approximated by the observed distribution (on our past lightcone),
    \item the resulting BH average volume can be approximated by the volume of the average BH. Given that the available distribution concerns BH masses, we used the volume of the average BH mass rather than an average volume which distribution we have no access at all.
\end{enumerate}
We evaluate the average BH volume by first getting the volume for a given mass, then using knowledge of the Universe's BH distribution to evaluate the average BH mass. The resulting volume of the average BH is then taken as BH average volume.

\subsubsection{Volume of a given BH}

The total volume of the boundary of a BH of given mass $M$ until its evaporation is evaluated considering it appears at creation with initial mass $M$ and evaporates through Hawking radiation \cite[e.g.][]{Wald:1984rg,Carroll:2004st}. The formation phase is neglected since its dynamical time is expected to be considerably much less than the formation time, and a simple isolated Schwarzschild BH is considered. The resulting volume read as 
\begin{align}
     \text{Vol}_{BH} \left( M \right)=&1.96\times 10^{87}\left(\frac{M}{M_\odot}\right)^5 m^3.\label{eq:BHBounVol}
\end{align}

\subsubsection{Evaluation of the average BH mass}

We are limited by our present knowledge of the BH mass distribution. We chose to evaluate it using Refs~\cite{Kovetz:2016kpi,Mutlu-Pakdil2016,Garcia-Bellido:2017fdg,Christian:2018mjv}. We used those references to evaluate averages and variances for the masses of Stellar made, Primordial, Intermediate mass and Super Massive BHs (SMBHs). Despite being aware of the limits on knowledge about the BH population, we resolved to treat available information as informed views and indications on it and present the information we were able to extract as signs of the potential of the model. From the theoretical figure of \cite{Garcia-Bellido:2017fdg} (Fig 7, indication of shape and proportions), we expect the BH distribution to present 4 peaks: Stellar BHs (SBHs), Primordial BHs (PBHs), Intermediate mass BHs (IMBHs) and Super Massive BHs (SMBHs).

Fig 10 of \cite{Mutlu-Pakdil2016} evaluates the peak of SMBHs, while
figs 2 and 6 of \cite{Kovetz:2016kpi} assesses astrophysical BHs from gravitational waves (GW),
and can be combined with Fig.~1 from the theoretical study of \cite{Christian:2018mjv} to evaluate the peak of SBHs.

The SBHs part of the distribution can proceed from fitting both Fig.~1 from \cite{Christian:2018mjv} and fig. 6 of \cite{Kovetz:2016kpi}, with a power law distribution. Then, taking their geometrical mean while variance is evaluated from their maximum variations for both standard deviations, yields an evaluation with standard deviation as $M_S=10^{1.09^{+0.27}_{-0.48}}$. 

The SMBHs distribution can be assessed with Fig. 10 of \cite{Mutlu-Pakdil2016} as it appears as a double power law. An evaluation of that distribution mean and variance for SMBHs provides $M_{SM}=10^{11.69^{+0.20}_{-0.39}}$.

Since knowledge for the two remaining peaks does not exceed that presented in Fig 7 of \cite{Garcia-Bellido:2017fdg}, we used it as indication of the relative peak heights and standard deviations of their Gau\ss ian shapes, since the mass axis is providing numerical values, to obtain estimates for the remaining peaks for Primordial and Intermediate mass BH, $M_P=10^{1.96^{+0.14}_{-0.21}}M_\odot$ and $M_I=10^{4.19^{+0.14}_{-0.21}}M_\odot$. We are well aware that these estimates are not based on strict modeling nor observations but, since no possibility was given to evaluate the errors on this method, we considered the variances obtained as lower estimates of errors and continued to treat them as mere variances.

As those peaks span several orders of magnitude in mass, we estimated the mean mass from these evaluations by combining them together in a weighted geometric mean. Putting together those evaluations with a geometrical weighted average and conservative treatment of variance yields $\langle M \rangle= \left(M_S^{1.63} M_P^{2.58} M_I M_{SM}^{1.32}\right)^\frac{1}{6.53}$, an estimate of the average BH mass
\begin{align}
    \left\langle M_{BH}\right\rangle \sim & 10^{4.04^{+0.49}_{-0.61}}M_\odot. \label{eq:AvgBHmassU}
\end{align}

\subsubsection{boundary volume}

Putting together the average mass estimate (\ref{eq:AvgBHmassU}) into the BH boundary volume (\ref{eq:BHBounVol}), one gets the estimate for the Universe's boundary volume
\begin{align}
    \text{Vol} \left( \partial \mathcal{M} \right) \sim &  \mathcal{N}_{BH}10^{107.5^{+2.5}_{-3.1}} m^3.\label{eq:BounVol}
\end{align}
\subsection{Topological cosmological constant}
Combining the expression (\ref{eq:Topolambdabody}) with the evaluation of the average volume through the boundary volume (\ref{eq:BounVol}), we get the final result
\begin{align}
\Lambda \approx 10^{-52.9_{-1.3}^{+1.5}} \, k_E\mathcal{C},\label{eq:TopoLambdaKC}
\end{align}
which can be further approximated using $k_E\approx1$ into
\begin{align}
\Lambda \approx 10^{-52.9_{-1.3}^{+1.5}} \, \mathcal{C}.\label{eq:TopoLambdaWithC}
\end{align}

\subsubsection{Confrontation with observed cosmological constant}

From Planck observations \cite{Planck2018}, we currently evaluate $\Lambda_O=10^{-51.08\pm0.01}m^{-2}$. In the case of a 4-manifold with null sectional curvature, Ref.~\cite{ASENS, Hoffman} evaluated the isoperimetric constant $\mathcal{C} = 11.8$. Using it as an evaluation for the topological cosmological constant one obtains
\begin{align}
 \Lambda \approx 10^{-51.8^{+1.5}_{-1.3}} m^{-2},   
\end{align}
which is compatible with the observed value. The topology of the Universe can therefore be considered as a serious candidate for the cosmological constant, without raising the usual cosmological constant fine tuning problem.

\subsubsection{Evaluation of the Universe isoperimetric constant}

Since the isoperimetric constant of the Universe remains unknown, although its order of magnitude is expected in the range of $\approx 10$, this evaluation (\ref{eq:TopoLambdaWithC}) through the average BH mass can be combined with cosmological constant observations \cite{Planck2018} to appraise its actual value, and therefore characterise the topology of the Universe with
\begin{align}
    \mathcal{C}=10^{1.82^{+1.31}_{-1.51}}.
\end{align}

\section{\label{sec:Conclusion}Conclusion}

In this work, we have discussed the application to the cosmological constant problem of the topological invariants examined in \citep{Lorca}, introduced as Lagrange multipliers, that are compatible with GR for cosmological manifolds. At the level of the Einstein-Hilbert action, this results in the emergence of a topological effective cosmological constant from the Euler number. To do so we restricted an Einstein-Cartan theory to solutions with GR-like behavior at the level of the field equations, a sufficient but not necessary condition to recover GR. As we linked the torsion to the topological invariants, the extremely strong attractor behaviour of the invariant solutions, fixing the manifold topology, allows to consider gravity as GR for most of the observable Universe.

The resulting effective topological cosmological constant (TCC), assuming the bare constant is set to zero by symmetry considerations, can then be computed as founded on the inverse square root of the volume of the manifold boundary, the volume of all black holes in the Universe history, and therefore yields an evaluation, that carries the caveats of our current knowledge on BH distribution in the universe, encompassing the observed value of the cosmological constant \citep{Hinshaw:2012aka,Bennett:2012zja,Aghanim:2012vda,Ade:2013sjv,Planck2018}. That value can then be considered seriously as sourced by the topology of the Universe, in particular its Euler number. This also solves the fine tuning problem of the cosmological constant \citep{Weinberg:1988cp,Weinberg:2008zzc}, since the inverse square root of the volume of the boundary of the Universe is quite naturally much smaller than the Planck density, and links the coincidence problem with the distribution of BHs in the Universe. Moreover, we used the observed value of the cosmological constant, combined with our evaluation, to provide a first evaluation of the isoperimetric constant $\mathcal{C}^2$ of the Universe, that is the ratio of the compact volume domain over its boundary hypersurface volume. In light of these results, we argue that the dynamical coupling zero-forms act as a kind of topological quintessence field, effecting the accelerated expansion of the Universe to comply with its BH boundary.

As this result proceeds from BH distribution evaluation that involved GW input, we can be confident that future GW detections and other improvement in our understanding of the BH population should improve our determination of the TCC and of the Universe's isoperimetric constant. A dynamic theory of the topology of the universe in the line of emerging geometry \citep{Brandenberger:2008nx,LevasseurPerreault:2011mw,Brandenberger:2012um,Afshordi:2014cia} should provide the framework for understanding this choice of $\mathcal{C}^2$ and perhaps for solving as well the coincidence problem
\citep{Amendola:1999er,TocchiniValentini:2001ty,Zimdahl:2001ar,Zimdahl:2002zb} by producing a mechanism of selection for the Betti numbers of the space-time manifold.

\section*{Acknowledgments}

The authors wish to thank M.Fontanini and E. Huguet for very useful discussions. The work of M.Le~D. has been supported by Lanzhou University starting fund and PNPD/CAPES20132029. M.Le~D. also wishes to acknowledge IFT/UNESP for hosting the beginning of this project. 

%% The Appendices part is started with the command \appendix;
%% appendix sections are then done as normal sections
\appendix

\section{\label{sec:BHsurf}Black Hole boundary of the universe}

Since we have made the physically reasonable assumption that a cosmological space-time boundary should be identified mostly with its $3$-punctures (the third Betti number $b_3 \left( \mathcal{M} \right)$), i.e the number of singularities it contains. These are, in turn, essentially the number of black holes ($\mathcal{N}_{BH}$, where BH stands for black hole from now on). Thus, we have the approximation:
\begin{align}
b_3 \left( \mathcal{M} \right) \simeq  \mathcal{N}_{BH} \quad .
\end{align}

Hence, the estimation of the boundary of the universe's manifold can be obtained systematizing in the following way:
\begin{enumerate}
    \item the boundary is mostly made of BHs,
    \item the hypersurface of the boundary is made of the sum for all BH of their integrated horizon surface, this is:
    \begin{gather}
        \text{Vol} \left( \partial \mathcal{M} \right)=\sum_{BH}\int_{\text{BH life}}d\tau 4\pi R_{Sch}^2,
    \end{gather}
    \item the third betti number $b_3$, thought of as the number of $3$-punctures in the manifold, is basically equal to the number of BH: $b_3 = \mathcal{N}_{BH}$ (see Eq.~\ref{nEtop})
\end{enumerate}
We can then deduce the average boundary volume factor to follow from
\begin{align}
    \left\langle \text{Vol} \left( \partial \mathcal{M} \right) \right\rangle= & \frac{1}{\mathcal{N}_{BH} \sim b_3}\sum_{BH}\int_{\text{BH life}}d\tau 4\pi R_{Sch}^2 \nonumber\\
    \simeq & \int_{\left\langle \text{BH}\right\rangle \text{ life} }d\tau 4\pi R_{Sch}^2.\label{eq:AvgBounVol}
\end{align}

Thus, such average represents the volume enclosed by the horizon of an average mass BH, since its manifold causally disconnected interior can be considered as not part of the Universe.

We will start by calculating a given BH volume as a function of its mass and then apply it to the average BH mass. To simplify, we take a Schwarzschild BH, neglect its collapse time and integrate its horizon over the time it needs to decay through Hawking radiation.

\subsection{BH volume}
As can be seen in \cite[p418]{Carroll:2004st}, Hawking radiation leads to BH evaporation in a finite time estimated as
\begin{align}
    \tau_{BH}=& \left(\frac{M}{M_\odot}\right)^310^{71} s\\
    \Rightarrow c\tau_{BH}=& 3\left(\frac{M}{M_\odot}\right)^310^{79} m,
\end{align}
so the surface of the boundary from a BH of creation mass $M$ is of order its surface $S_{BH}$
times its development till evaporation $c\tau_{BH}$. The surface being evaluated, approximating the BH to a Schwarzschild one, with the Schwarzschild radius
\begin{align}
    r_{S}=& 2\frac{GM}{c^2},
\end{align}
we can find it, the solar mass Schwarzschild radius being
\begin{align}
    r_{S\odot}=& 2\frac{GM_\odot}{c^2}=2950m,
\end{align}
 in the form
\begin{align}
    S_{BH}=& 4\pi r_{S}^2=4\pi r_{S\odot}^2\left(\frac{M}{M_\odot}\right)^2\nonumber\\
    =&1.09\times 10^8\left(\frac{M}{M_\odot}\right)^2 m^2,
\end{align}
and thus estimate the order of magnitude of the volume of the boundary from one BH as
\begin{align}
    V=& S_{BH}\times c\tau_{BH}\nonumber\\
    =& 3.28\times 10^{87}\left(\frac{M}{M_\odot}\right)^5 m^3.
\end{align}

Now, as can be seen in \cite[p412]{Wald:1984rg}, this estimate comes from
\begin{align}
    \frac{dM}{d\tau}\propto & -\frac{1}{M^2}\nonumber\\
    =&-\frac{1}{C M^2}\\
    \Leftrightarrow c d\tau_{BH}=&-c\, C M^2 dM.
\end{align}
Since the time estimate is given by the integral
\begin{align}
    \tau_{BH}=& \int_M^0 \left(-C M^2\right) dM = \frac{C}{3}M^3\\
    \Rightarrow C=&\frac{3\times 10^{71}}{M_\odot^3},
\end{align}
we can then evaluate the total actual volume of the boundary from one BH as
\begin{align}
    V=& \int_{\tau_i}^{\tau_f} S_{BH}\times cd\tau_{BH}\nonumber\\
    =& \int_M^0 \left(-S_{BH}c\frac{3\times 10^{71}M^2}{M_\odot^3}\right)dM\nonumber\\
    =& \int^M_0 9.81\times 10^{87}\frac{M^4}{M_\odot^5}dM\nonumber\\
    =& 1.96\times 10^{87}\left(\frac{M}{M_\odot}\right)^5 m^3.\label{eq:1BHvol}
\end{align}

\subsection{Average BH}

Examining the theoretical figure of \cite{Garcia-Bellido:2017fdg} (Fig 7, indication of shape and proportions),
Fig 10 of \cite{Mutlu-Pakdil2016},%\\\href{http://d.umn.edu/~msseigar/papers/MutluPakdil2016.pdf}{\uline{found here}},
for Super Massive BHs (SMBHs),
figs 2 and 6 of \cite{Kovetz:2016kpi} for astrophysical BHs from gravitational waves (GW),
and from \cite{Christian:2018mjv} %,seen in \\\href{https://aasnova.org/2018/05/09/sizes-of-black-holes-throughout-the-universe/}{\uline{there}},
the figures can be used to evaluate a typical average BH.

%Taking the Milky Way (MW) as a roughly typical galaxy (BHs of the MW as good representatives), we add to the distribution the central SMBH mass of order $10^6 M_\odot$ to roughly calculate the average mass with the convergence distribution of \cite{Christian:2018mjv}.
Fig.~1 from the theoretical study of
\cite{Christian:2018mjv}, combined with the prediction on GW detection from \cite{Kovetz:2016kpi}, fig. 6, allow an evaluation of the stellar made BH distribution by fitting both with a power law distribution and taking their geometrical mean while variance is evaluated from their maximum variations for both standard deviations. The evaluation yields $M_S=10^{1.09^{+0.27}_{-0.48}}$. Fig. 10 of \cite{Mutlu-Pakdil2016} yields the mean and variance evaluation of that distribution for SMBHs, and, combined with fig.~7 from \cite{Garcia-Bellido:2017fdg}, allows for an evaluation of the remaining peaks for Primordial and Intermediate mass BH. The results give  $M_{SM}=10^{11.69^{+0.20}_{-0.39}}$, $M_P=10^{1.96^{+0.14}_{-0.21}}M_\odot$ and $M_I=10^{4.19^{+0.14}_{-0.21}}M_\odot$. Putting all these evaluations together in a weighted geometric mean, $\langle M \rangle= \left(M_S^{1.63} M_P^{2.58} M_I M_{SM}^{1.32}\right)^\frac{1}{6.53}$, allows one to estimate the average and standard deviation-generated variances of the typical BH mass, estimating variances with propagation and uncertainty on our evaluations,

\begin{align}
    \left\langle M_{BH}\right\rangle \sim & 10^{4.04^{+0.49}_{-0.61}}M_\odot \label{eq:AvgBHmass}
\end{align}

and use it to evaluate the corresponding %order of magnitude 
value of the topological cosmological constant from Eq. (\ref{eq:Topolambda}).

\subsection{Topological cosmological constant}
Combining Eqs.~(\ref{eq:AvgBHmass}) and (\ref{eq:1BHvol}) one gets the average boundary volume
\begin{align}
    \left\langle \text{Vol} \left( \partial \mathcal{M} \right) \right\rangle\sim & 1.96 \times 10^{87}\left(10^{4.04^{+0.49}_{-0.61}}\right)^5 m^3 \nonumber\\
    \sim & 10^{107.5^{+2.5}_{-3.1}} m^3,
\end{align}
which in turns yields the topological cosmological constant from Eq. (\ref{eq:Topolambda}), using $k_E\approx1$,
\begin{align}
\Lambda \approx 10^{-52.9_{-1.3}^{+1.5}} \, \mathcal{C}\label{eq:TopoLambdaC}
\end{align}

It seems that the isoperimetric constant $\mathcal{C}$ is in the order of $\sim 10^1$ See for instance \cite{ASENS, Hoffman} with a calculated $\mathcal{C} = 11.8$ for a 4-manifold with null sectional curvature\footnote{The calculation goes as follows: $\mathcal{C} = \left\{ 2^3 \cdot \frac{4 \pi}{3}  \right\}^{\frac{1}{2}} \cdot \left( 2 \cdot \omega_4 \right)^{-\frac{1}{8}} \simeq 11.8 $ , where $\omega_4$ is the volume of the four dimensional unit sphere.}. Then:
\begin{align}
 \Lambda \approx 10^{-51.8^{+1.5}_{-1.3}} m^{-2}    
\end{align}

which is compatible with the value 
\begin{align}
    \Lambda_O=&(4.24 \pm 0.11) \times
10^{-66} eV^2 \nonumber\\
=&\left(1.97\times 10^{-7}\right)^{-2}(4.24 \pm 0.11) \times
10^{-66} m^{-2} \nonumber\\
=&\left(8.35 \pm 0.22\right)\times 10^{-52} m^{-2} \nonumber\\
=&10^{-52+0.92\pm0.01}=10^{-51.08\pm0.01}
\end{align}
observed \cite[e.g.][]{Planck2018}. This allows us to conclude that topology can explain the order of magnitude and value of the cosmological constant.

As $\mathcal{C}$ is actually unknown for our Universe, the topological evaluation of the cosmological constant (\ref{eq:TopoLambdaC}) can be combined with observations from \cite{Planck2018} to obtain its measurement: our Universe's topology can be characterised with its isoperimetric constant
\begin{align}
    \mathcal{C}=10^{1.82^{+1.31}_{-1.51}}.
\end{align}

\section{\label{sec:cartan} Formalism overview}

For the following sections and appendices let $\Omega^k \left( M \right)$ be the space of smooth complex exterior differential forms of degree $k$ over a $4$-manifold $M$ and $\Omega^k \left( M \right)^* $  its dual. Let $\Omega \left( M \right) = \oplus^4_{k=0} \Omega^k \left( M \right) $ its graded algebra. The usual operators on $\Omega \left( M \right)$ are well defined: 
\begin{description}[]
    \item[Exterior derivative] \hfill \\ the differential $d: \Omega^k \left( M \right) \rightarrow \Omega^{k+1} \left( M \right)$ with $0 \leq k \leq 3$,
    \item[Exterior co-derivative] \hfill \\ the co-differential $d^*: \Omega^k \left( M \right) \rightarrow \Omega^{k-1} \left( M \right) $ with $1 \leq k \leq 4$ and
    \item[Hodge dual] \hfill \\ the dual $\star: \Omega^k \left( M \right) \rightarrow \Omega^{4-k} \left( M \right)$ with $0 \leq k \leq 4$.
\end{description}

\subsection{\label{subsec:VEPformalism} Vielbein-Einstein-Palatini (VEP) formalism}

The standard General Relativity (GR) is derived from the Einstein-Hilbert action, expressed in terms of the metric and the Ricci tensor. However, the vacuum GR can equivalently proceed from the Vielbein-Einstein-Palatini (VEP) action \cite{Dadhich}. The field equations are then obtained through the Einstein-Cartan formalism, for which the basic geometrical ontologies are shown in table~\ref{tab:VEP-geometrical-elements}\footnote{We are here introducing the notation $A^{\left( * \right)}_{ab} : = \frac{1}{2} \epsilon_{abcd} A^{cd}$ for the Lie dual acting over any $A^{cd} \in \Omega \left( \mathcal{M} \right)$ with two spin indices $c,d$.}. 
\begin{center}
\begin{table}
\begin{centering}
\begin{tabular}{c|c}
\hline 
$e^{a}\in\Omega^{1}\left(M\right)$  & $1$-form frame basis \tabularnewline
 & ( or vierbien ) \tabularnewline
$e_{a}\in\Omega^{1}\left(M\right)^{*}\simeq T\left(M\right)$  & $1$-vectors frame basis \tabularnewline
 & ( or dual vierbein ) \tabularnewline
$g_{ab}:=g\left(e_{a},e_{b}\right)\in\Omega^{0}\left(M\right)$  & metric tensor \tabularnewline
 & components \tabularnewline
$\hat{\Sigma}^{ab}:=\frac{1}{2}e^{a}\wedge{e}^{b}\in\Omega^{2}\left(M\right)$  & Palatini $2$-form \tabularnewline
 & ( Normalized ) \tabularnewline
$d\mu:=\frac{1}{3}\hat{\Sigma}^{ab}\wedge\hat{\Sigma}_{ab}^{(*)}\in\Omega^{4}\left(M\right)$  & the volume $4$-form \tabularnewline
 & ( $\star\left(1\right):=d\mu$ ) \tabularnewline
$\bar{\omega}_{\;\;b}^{a}\in\Omega^{1}\left(M\right)$  & the connection $1$-form \tabularnewline
 & ( Levi - Civita ) \tabularnewline
\hline 
\end{tabular}
\par\end{centering}
\caption{\label{tab:VEP-geometrical-elements}VEP geometrical elements}
\end{table}
\par\end{center}

In Einstein-Cartan formalism, the curvature $2$-form $\bar{R}^a_{\;\;b} \in \Omega^2 \left( {M} \right)$ and the torsion $2$-form $\bar{T}^a \in \Omega^2 \left( M \right)$ are defined as: :
\begin{align}
\label{originalcurvature} 
\bar{R}^{a}_{\;\;b} &:= d \bar{\omega}^{a}_{\;\;b} + \bar{\omega}^a_{\;\;c} \wedge \bar{\omega}^{c}_{\;\;b} \;\;, \\
\label{originaltorsion} \bar{T}^a &:= d_{\bar \omega} e^a = d e^a + \bar{\omega}^{a}_{\;\,b} \wedge e^b \;\;,
\end{align}
where $d_{\bar \omega}$  denotes the exterior covariant derivative with respect to the connection $1$-form $\bar \omega^{a}_{\;\;\,b}$. The field equations are obtained after the variations of the following gravitational action, defined over a closed manifold $M$ ( i.e $\partial M = \emptyset$ ) : 
\begin{align}
\nonumber \small{ {S_{G}}[{e}^c, \bar{\omega}^a_{\;\;b}, \Lambda] } & := \int\limits_{M} \frac{1}{\kappa} \mathfrak{L}_G \quad \quad \quad | \quad \quad \mathfrak{L}_G \in \Omega^4 \left( M \right) \\
\label{EHaction} & = \int\limits_{M} \frac{1}{\kappa} \left( \bar{R}^{\left( * \right)}_{ab} \wedge \hat{\Sigma}^{ab} - \Lambda {d \mu} \right) \;\;,
\end{align}
 by considering the connection $\bar{\omega}^a_{\;\;b}$ and the vierbein\footnote{Recall that a vielbein frame, for an n-dimensional manifold $M$ designate $e^a$ either as vielbein (viel: german, many) in any dimension, or vierbein (vier: german, four), or tetrad, in the $n=4$ case.} $e^a$ as independent. In the latter, $\kappa = 8 \pi G c^{-4}$ is the Gravitation constant and $\Lambda$ is a cosmological constant-type term. At this point, it can take any value. Later on it can describe the source for accelerated expansion. Variation of (\ref{EHaction}) with respect to the vierbein ($\delta e^a$) gives the Einstein's field equations, while the variation with respect to the connection ($\delta \bar{\omega}^{ab}$) gives the null torsion condition. The last field equation, particularly useful, sets the connection to be the Levi-Civita through the relation:
\begin{equation}\label{LeviCivita}
\bar{\omega}^{a}_{\;\;\,b} = -\frac{1}{2} \left[ i^a \left(d e^{\flat}_b\right) - i_b \left( d e^a \right) + i^{a}_{\;\;\,b} \left( d e^{\flat}_c \right)e^c \right] \;\;,
\end{equation}
where $i_a := i_{e_a} : \Omega^k \left( M \right) \rightarrow \Omega^{k-1} \left( M \right)$ is the slant, interior product or contraction with respect to the dual basis $e_a$. By abuse of notation we are also writing the analogous operation induced by the metric as $i^a := g^{-1}_{ab} i_b := g^{ab} i_b $. We have also used the so called musical isomorphisms between differential forms and tangent vectors $\flat: \Omega^k \left( M \right)^* \rightarrow \Omega^k \left( M \right)$ and $\sharp: \Omega^k\left( M \right) \rightarrow \Omega^k \left( M \right)^* $.

The action (\ref{EHaction}) is usually understood as mostly a formal device to obtain the local field equations. We aim to extend this formalism in two ways: 
\begin{enumerate}[a)]
    \item We allow for the manifold to be compact but not necessarily closed. We will denote this change by  $M \mapsto \mathcal{M} $ with a non trivial boundary $\partial \mathcal{M} \neq \emptyset$.;
    \item We need this modified action to define well posed variational expressions. We mainly understand this well posed-ness naively as a physically sound and unambiguously defined field equations;
\end{enumerate}
 It can be easily proven that the first term in the action (\ref{EHaction}) does not contribute with boundary terms. However, the second must be taken carefully. This situation is taken into account in Sec.~\ref{subsec:wellposed} where we find the right modifications for the action to be well posed.

\subsection{\label{TopManifold} Characteristic classes and some topology}

Characteristic classes can be described as global invariants that measure the deviation of a local product structure from a global product structure. On an oriented space-time manifold $\mathcal{M}$ taken as a compact $4$-manifold,
the characteristic classes of the tangent bundle available and consistent with the space-time topology are \cite{Sengupta,Kaul,donaldson,Babourova,Zanelli,Chandia,Liko}: 
\begin{description}[]
    \item[The Pontryagin class] \hfill \\  $ C_P := p_1 \left( T \mathcal{M} \right) \in H^4 \left( \mathcal{M} ; \mathbb{Z}\right) \simeq \mathbb{Z} $,
    \item[The Euler class] \hfill \\ $C_E := e \left( T \mathcal{M} \right) \in H^4 \left( \mathcal{M} ; \mathbb{Z}\right) \simeq \mathbb{Z}$, and
    \item[The Chern class] \hfill \\ $C_{N} := c_2 \left( T \mathcal{M} \right) \; \in H^{4} \left( \mathcal{M} ; \mathbb{Z} \right)$ (Nieh-Yan).
\end{description}
They can be written as exact differentials of $3$-forms, or equivalently as the $4$-forms:
\begin{align}
\label{Pontryagin} \small C_{P}  & =  \frac{1}{8 \pi^2} d \left( \omega^a_{\;\;b} \wedge \left[ R^b_{\;\;a} - \frac{1}{3} \omega^b_{\;\;c} \wedge \omega^c_{\;\;a} \right] \right) \\ 
\nonumber & = \frac{1}{8 \pi^2} R^a_{\;\;b} \wedge R^b_{\;\;a} \;\;,\\
\label{Euler} \small C_{E} & = \frac{1}{16 \pi^2}  d  \left(\epsilon_{abdc} \omega^{ab} \wedge \left[ R^{cd} - \frac{1}{3} \omega^c_{\;\;f} \wedge \omega^{fd}\right] \right) \\ 
\nonumber & = \frac{1}{8 \pi^2} R^{ab} \wedge R^{\left( * \right)}_{ba} \;\;, \\
\nonumber \small C_N & = \small d \left( e^a \wedge T_a \right) \\ 
\label{NiehYan} & = T^a \wedge T_a - R_{ab} \wedge \Sigma^{ab} \;\;,
\end{align}
respectively, where we have raised and lowered indices by means of the metric. While the first two classes are dimensionless and only involve the connection, $C_{N}$ has dimensions of $\left[ \text{length}^2 \right] $ and depends explicitly on the torsion and, is thus null in the absence of it.

The following is an important standard result that we give without proof but will be used in the subsequent Appendices:
\begin{thm}\label{topnumber}
Denote the integration of the characteristic densities by $n_j := \int_{\mathcal{M}} C_j$, for any simply connected oriented smooth $4$-manifold, $n_j \in \mathbb{Z}$  (for $j = E , P$) \cite{milnor,donaldson} .
\end{thm}

\begin{rem}
Since we are allowing $\Omega \left( \mathcal{M} \right)$ to be the graded algebra of smooth \textit{complex} exterior differential forms, we are forced to modify the last result to read:
\begin{align}\label{topnumbers}
n_j = \mathfrak{Re} \int\limits_{\mathcal{M}} C_j \; \in  \; \mathbb{Z} \;\; , \;\, \text{with} \;\; j \in J \;\; ,
\end{align}
which we subsequently call \textit{the topological numbers}.
\end{rem}

We also provide here the informal definition of the $n^{th}$ Betti numbers of a manifold $\mathcal{M}$, which are the rank of the $n^{th}$ \textit{homology group} $H_n \left( \mathcal{M} ; \mathbb{R} \right)$. This is:
\begin{align}\label{betti}
b_n := \dim \left( H_n \left( \mathcal{M} ; \mathbb{R} \right) \right) \;\, , \,\; n = 0,1, \cdots, \dim \mathcal{M} \;\;.
\end{align}
These are also topological invariants of $\mathcal{M}$ and, intuitively, they represents the number of $n$-dimensional \textit{punctures or holes on a topological space}.

\subsection{Contortion and the first order formalism}

When the characteristic classes are included in the action, their presence will allow for a theory of gravity with non-null torsion. Furthermore, since torsion becomes non trivial when considering the above topological terms, we make use of the \textit{first order formalism}  \cite{Mielke,Cho,Hatzinikitas,Cambiaso,Leigh,Petkou,Randal} to take its effects into account. Here, the contribution from the new terms added at the level of the action results in a second term added to the Levi-Civita connection. Thus, the new connection $1$-form and torsion yield:
\begin{align}
\label{connection} \omega^{a}_{\;\;b} & := \bar \omega^{a}_{\;\;b} + K^{a}_{\;\;b}  & \in \Omega^1 \left( \mathcal{M} \right) \; ,\\
\label{contortion} d_{\omega} e^a & := T^a = K^{a}_{\;\;b} \wedge e^b & \in \Omega^2 \left( \mathcal{M} \right), 
\end{align}
where $K^{a}_{\;\; b}$ is known as the \textit{contortion} $1$-form, which in turn can be inverted using the slant product to write:
\begin{equation}\label{contortion1}
{K^a_{\;\;\,b}} =  - \frac{1}{2}\left\{ {{i^a}\left( {{T_b}} \right) - {i_b}\left( {{T^a}} \right) - {i^a_{\;\;\,b}}\left( {{T_c}} \right) \wedge {e^c}} \right\} \;\;.
\end{equation}

I%t is i
n this sense, %that 
the contortion $K^{a}_{\;\;\, b}$ is responsible for the appearance of torsion $T^a$. Similarly, the curvature $2$-form can %thus 
be written in the following way
\begin{align}\label{Riemanntensorgral}
R^a_{\;\;b} =  \bar{R}^a_{\;\;b} +  \Theta^a_{\;\;b}  \quad \in \Omega^2 \left( \mathcal{M} \right) \;\;,
\end{align}
where Eqs.~(\ref{originalcurvature}) and (\ref{connection}) have been used and the term
\begin{align}\label{torcurv}
\Theta^a_{\;\;b} := d_{\omega} K^a_{\;\;b} - K^a_{\;\;c} \wedge K^{c}_{\;\;b} \quad \in \Omega^2 \left( \mathcal{M} \right) \;\;,
\end{align}
is the contortion related part of the total curvature. To complete the picture, the Bianchi identities
\begin{align}
\label{Bianchi} d_{\omega} T^a = R^a_{\;\;b} \wedge e^b \;\; , \; \; d_{\omega} R^a_{\;\;b} = 0 \;\;,
\end{align}
hold for any $R^a_{\;\;b}$ and $T^a$.

The first order formalism then assumes this new connection $1$-form $\omega^a_{\;\;b}$ as well as the vierbein $e^a$ to be independent quantities. The independent variation of the latter yield the field equations of the theory.

\subsection{Further notation and terminology}

Since we have assumed that our space-time manifold $\mathcal{M}$ is compact but not closed, i.e. with a non-trivial boundary $\partial \mathcal{M} \neq \emptyset$. We begin by recalling that, given the Hodge dual\footnote{The hodge dual operator is defined over the veirbein basis as $\star \left( e^{a_1} \wedge \cdots \wedge e^{a_n} \right) := \frac{1}{\left( 4 - n\right)!} \epsilon^{a_1 \cdots a_n}_{\;\;\;\;\;\; \;\;\;\;\;a_{n+1} \cdots a_4} e^{a_{n+1}} \wedge \cdots e^4$ and it extends to the entire space $\Omega \left( \mathcal{M} \right)$ by linearity.} $\star: \Omega^k \left( \mathcal{M} \right) \rightarrow \Omega^{n-k} \left( \mathcal{M} \right)$, the operations ($s = -1 $)
\begin{align}\label{hodge}
\star \star = \left( - 1 \right)^{k\,(n-k)} s \;\;, \quad & \quad \star d^{\star} = \left( - 1 \right)^k s d \star \;\;, \\
\nonumber \star d = \left( - 1 \right)^{k+1} s d^{\star} \;\;\; &\text{on} \;\;\; \Omega^k \left( \mathcal{M} \right) \;\;,
\end{align}
are well defined. Therefore, the $L^2$-inner product induced by the Hodge dual for all $\alpha , \beta \in \Omega^k \left( \mathcal{M} \right)$ is naturally defined in the following way:
\begin{align}
\label{L2inner} & \left( \alpha , \beta \right) = \int\limits_{\mathcal{M}} \alpha \wedge \star \beta = \int\limits_{\mathcal{M}} \left\langle \alpha , \beta \right\rangle d \mu \;\; ,
\end{align}
while the $L^2$-norm in $\mathcal{M}$ will be written as:
\begin{align}\label{L2norm}
&\left\| \alpha \right\|^2_{L^2} = \left( \alpha , \alpha \right) = \int\limits_{\mathcal{M}} \left\langle \alpha, \alpha \right\rangle d \mu \;\;,
\end{align}
where, the scalar product $\left\langle \cdot , \cdot \right\rangle : \Omega^1 \left( \mathcal{M} \right) \times \Omega^1 \left( \mathcal{M} \right) \rightarrow \mathbb{C}$ is defined to be the tetrad metric $g^{ab} := \left\langle e^a , e^b \right\rangle$ for the basis, taken to be skew-symmetric as usual and can be naturally extended to $n$-forms by linearity. Hence, for any two $1$-forms $\alpha = \alpha_a e^a$ , $ \beta = \beta_b e^b \in \Omega^1 \left( \mathcal{M} \right)$, we then have:
\begin{align}\label{scalarproduct}
\left\langle \alpha , \beta \right\rangle = \left\langle \beta , \alpha \right\rangle^* = g^{ab} \alpha_a \beta_b \;\; ,
\end{align}
with $\left( \cdot \right)^* : \Omega^0 \left( \mathcal{M} \right) \rightarrow \Omega^0 \left( \mathcal{M} \right)$ being the complex conjugate operation. This naturally leads us to consider the following metric compatibility conditions:
\begin{align}\label{metriccomp}
g^{ab} g_{ab} = \delta^a_{\;\;b} \;\; \text{and} \;\; d_{\omega} \left( g_{ab} \right) = 0 \;\; ,    
\end{align}
allowing us the rising and lowering of spin indices.

Under the same context, the Lie derivative over differential forms is defined as usual:
\begin{align}\label{Liederiv}
\mathcal{L}_{a} \omega := i_{a} d \omega + d i_{a} \omega \;\; , \;\; \omega \in \Omega \left( \mathcal{M} \right)
\end{align}
while, induced by the metric, we can also define:
\begin{align}\label{Liederivdual}
\mathcal{L}^{a}  \omega := g^{ab} \mathcal{L}_b \omega \;\; , \;\; \omega \in \Omega \left( \mathcal{M} \right) \quad.
\end{align}

On the other hand, we will also refer to the following definition in order to obtain some analytic solutions for the curvature:
\begin{defn}[Einstein-type symmetry]\label{Einstein}
A $\Pi_{ab} \in \Omega^2 \left( \mathcal{M} \right)$ such that $\Pi_{ab} = - \star \Pi^{(*)}_{ab}$ is called Einstein-type $2$-form.
\end{defn}
Hence, according to the latter, the following immediate results follow:
\begin{rem}\label{curvA}
The forms $u\hat{\Sigma}^{\flat}_{ab}$ and $v\hat{\Sigma}^{(*)}_{ab}$ , with $u,v \in \Omega^0 \left( \mathcal{M} \right)$ are Einstein-type $2$-forms\footnote{In fact, for $u= v= \Lambda$, the cosmological constant, these solutions are the so called \textit{topological instantons} in a space-time manifold $\mathcal{M}$ homeotopic to $S^4$ \cite{Chandia,Chandia1}.}.
\end{rem}

\section{\label{sec:The model} The dynamical model}

We start by analyzing the important case of the finite volume to later consider a more general case and discuss the well posed-ness of the field equations obtained by the model.

\subsection{\label{subsec:comact} The case of the finite volume space-time manifold with boundary}

This subsection closely follows the material appearing in \cite{BELISHEV2008128}. Given the structure of $\mathcal{M}$ as defined in (\ref{hodge}), the orthogonal decomposition $ \Omega^k \left( \mathcal{M} \right) = \mathcal{E}^{k-1}_D \left( \mathcal{M} \right) \oplus \mathcal{C}^{k+1}_N \left( \mathcal{M} \right) \oplus  \mathcal{H}^k \left( \mathcal{M} \right)$, known as the Hodge decomposition, holds. Here $\mathcal{E}^{k-1}_D \left( \mathcal{M} \right)$ is the set of exact forms of order $k$, $\mathcal{C}^{k+1}_N \left( \mathcal{M} \right)$ is the set of co-exact forms of order $k$ and $ \mathcal{H}^k \left( \mathcal{M} \right) $ is the set of harmonic forms of order $k$.

Furthermore, as the manifold $\mathcal{M}$ has a non-null boundary $\partial \mathcal{M}$, the harmonic forms can be further decomposed as $\mathcal{H}^k = \mathcal{H}^k_D \oplus \mathcal{H}^k_{co}$ or $\mathcal{H}^k = \mathcal{H}^k_N \oplus \mathcal{H}^k_{ex}$, where $D$, $N$ stand for Dirichlet and Neumann, respectively, while $co$, $ex$ stand for co-exact and exact, respectively. In this setting, for any $0 \leq k \leq 4$, the Dirichlet to Neumann operator $\Lambda : \Omega^k \left( \partial \mathcal{M} \right) \rightarrow \Omega^{4-k} \left( \partial \mathcal{M} \right)$ is characterized as follows: Given $\chi \in \Omega^k \left( \partial \mathcal{M} \right)$,  the boundary value problem: 
\begin{align}\label{DNop}
\begin{cases}
\Delta \omega = 0 \;\;, \quad d^{\star} \omega = 0 \\ 
i^* \omega = \chi \;\;,  
\end{cases}
\end{align}
is solvable, where we have written $\Delta: \Omega^k \left( \mathcal{M} \right) \rightarrow \Omega^k \left( \mathcal{M} \right)$ for the Laplace-Beltrami operator. The solution $ \omega \in \Omega^k \left( \mathcal{M} \right) $ is unique up to an arbitrary Dirichlet harmonic field $\lambda_D \in \mathcal{H}^k_D \left( \mathcal{M} \right)$ and $i^{\star}: \Omega^k \left( \mathcal{M} \right) \rightarrow \Omega^{k} \left( \partial \mathcal{M} \right)$ is the inclusion in the boundary operator. Therefore, the Dirichlet to Neumann operator is the form $\Upsilon \left[ \chi \right] = i^{\star} \left( \star d \omega \right) = \left( -1 \right)^{k+1} i^{\star} \left( d^* \star \omega \right)$, a \textit{well defined} operator which is \textit{independent of the choice of the solution} $ \omega $. When particularizing for $k = 4$, we have:

\begin{prop}\label{DNoperator}
If the space-time manifold $ \mathcal{M} $ has a non-null boundary $\partial \mathcal{M}$, the cosmological constant $\Lambda$ is proportional to the Dirichlet to Neumann operator $\Upsilon \left[ \chi \right]$.
\end{prop}

\begin{proof}
The proof is a re-working of the proof appearing in \cite{BELISHEV2008128} for an $n$-manifold. We know that the space $\mathcal{H}^4 \left( \mathcal{M} \right)$ of harmonic fields of highest degree consists of forms $K d \mu $, where $K = const$. It follows that the cosmological constant term in (\ref{EHaction}) is a harmonic form. Then, we can write $ \frac{\Lambda}{\kappa} d \mu = d \omega$, with $\omega$ a solution to the boundary value problem discussed above. Thus, we have $\star d \omega = \frac{\Lambda}{\kappa} $, hence $ \Upsilon \left[ \chi \right] = i^{\star} \left( \star d \omega \right) = \frac{\Lambda}{\kappa} $. From here, on one side we have:
\begin{align*}
\int\limits_{\mathcal{M}} \frac{\Lambda}{\kappa} \epsilon = \frac{\Lambda}{\kappa} \int\limits_{\mathcal{M}} \epsilon = \frac{\Lambda}{\kappa} \, \text{Vol} \left( \mathcal{M} \right) \;\; ,
\end{align*}
and on the other side, by Stokes theorem we have:
\begin{align*}
\int\limits_{\partial \mathcal{M}} i^{\star} \omega = \int\limits_{\partial \mathcal{M}} \chi \quad \text{hence,} \quad 
\Upsilon \left[ \chi \right] = \frac{1}{ \text{Vol} \left( \mathcal{M} \right)} \int\limits_{\partial \mathcal{M}} \chi 
\end{align*}
where the Dirichlet to Neumann operator $\Upsilon \left[ \chi \right] \in \Omega^0 \left( \partial \mathcal{M} \right)$ is constant.  Therefore, $\Lambda = \kappa \Upsilon \left[ \chi \right]$, a constant, as expected.
\end{proof}

Proposition \ref{DNoperator} gives a somewhat satisfactory answer to the fine tuning problem for a manifold with finite volume and non-trivial boundary. Here, it is understood that the order of magnitude of the topological constant can be attributed to a relation of the $\text{Vol} \left( \mathcal{M} \right) >> \int_{\partial \mathcal{M}} \chi $. This last observation leads us to upgrade the cosmological constant $\Lambda$ to a \textit{cosmological functional} $ \tilde{\lambda}$.

Furthermore, if the kernel of the operator $\Upsilon \left[ \chi \right]$ is known, lower bounds for the Betti numbers can be obtained, which further stresses the connection between this result and the topology of the manifold $\mathcal{M}$  \cite[see][for further details]{BELISHEV2008128}.

\subsection{\label{subsec:themodel} The model}

Throughout the paper we will use the index set $J = \left\{ N, P , E \right\} $ with $j \in J$, so as to collect the subscripts related to the characteristic classes Nieh-Yan, Pontryagin and Euler, respectively. We then focus on the following family of actions: 
\begin{align}
\label{action1} S & := S_{G} \left[e^c, \omega^a_{\;\;b} , \tilde{\lambda} \right] + S_{T} \left[e^c, \omega^a_{\;\;b}, \varphi_j \right] +  S_M \left[ \Psi \right] \;\;\; 
\end{align}
The first term is the pure gravitational action:
\begin{align}
\nonumber {S_{G}} \left[e^c, \omega^a_{\;\;b} , \tilde{\lambda} \right] & = \int\limits_{\mathcal{M}} \frac{1}{\kappa } \left( {R}^{\left( * \right)}_{ab} \wedge \hat{\Sigma}^{ab} - \tilde{\lambda} {d \mu} \right) \\ 
\label{EHactionmod} & + \quad \quad \quad \quad \int\limits_{ \partial \mathcal{M}} i_n \left( \frac{\tilde{\lambda} d \mu}{\kappa} \right) \,,
\end{align}
where $e^c$ are the tetrad, $\omega^a_{\;\;b}$ is the total connection and $\tilde{\lambda}$ is the yet undefined upgraded quantity $\Lambda \rightarrow \tilde{\lambda}$, the \textit{cosmological functional}. Here and subsequently we insert boundary terms \textit{\`{a} la} York-Gibbons-Hawkins in order to have a well posed variation (see discussion in Sec.~\ref{subsec:wellposed}) so that $i_n: \Omega^4 \left( \mathcal{M} \right) \rightarrow \Omega^{3} \left( \partial \mathcal{M} \right)$ is the interior product in the direction of the outward unit normal vector field $n$ on $\partial \mathcal{M}$.

The second term in the right hand side of Eq.~(\ref{action1}), is the topological action:
\begin{align}
\label{Stop}  S_{T} \left[e^c, \omega^a_{\;\;b}, \varphi_j \right] & := - \int\limits_{\mathcal{M}} \frac{i}{\kappa} \left( \varphi_j C_j \right) + \int\limits_{\partial \mathcal{M}} \frac{i}{\kappa} i_n \left( \varphi_j C_j \right)
\end{align} 
(summation over $j \in J$ is understood, except when indicated), where $i := \sqrt{-1}$ is inserted for later convenience. This action is composed of \textit{current-like} terms, where $\varphi_j \in \Omega^0 \left( \mathcal{M} \right)$ for $j \in J$ are the coupling zero-forms, having dimensions of $ \left[ \texttt{length}^{2} \right]$ for $j=E,P$ and being dimension-less for $j=N$. Finally, $C_j$ ($j \in J$) are the characteristic classes consistent with the space-time manifold $\mathcal{M}$, all of which have been defined in Eqs.~(\ref{Pontryagin} - \ref{NiehYan}).

The last term in Eq;~(\ref{action1}) stands for an explicit matter action. It is composed of the aforementioned coupling zero-forms $\varphi_j$ and the spinor fields $\Psi := \left\{ \bar{\psi} , \psi \right\}$, with $\bar{\psi},\psi \in \Omega^0 \left( \mathcal{M} \right) \otimes \mathbb{G}$ with $\mathbb{G}$ being the set of complex Grassmann numbers. Concretely, we have:
\begin{align}
\nonumber & S_{M} \left[ \Psi \right] := - \int\limits_{\mathcal{M}} \frac{2}{\kappa} \mathfrak{L}_M \left[ \Psi \right] + \int\limits_{\partial \mathcal{M}} \frac{2}{\kappa} i_n\left( \mathfrak{L}_M \left[ \Psi \right] \right) \;\; \text{with} \\
\label{SMat} & \mathfrak{L}_M \left[ \Psi \right] := \mathfrak{Re} \left\{ \bar{\psi} \gamma^a D_{\omega} \psi \right\} \wedge \star e^{\flat}_a - V \left( \Psi \right)  d \mu \;\; ,
\end{align}
i.e. basically a mass-less Dirac Lagrangian density with fermion potential $V \left( \Psi \right): \Omega^0 \left( \mathcal{M} \right) \otimes \mathbb{G} \times \Omega^0 \left( \mathcal{M} \right) \otimes \mathbb{G} \rightarrow \mathbb{R}$. The latter takes $\gamma^a$ matrices as being a representation of the Clifford algebra $\left\{ \gamma^a , \gamma^b \right\} = 2 \eta^{ab} $ (with $\eta^{ab}$ the Minkowskian metric) and $\bar{\psi }:= \psi ^{\dagger} \gamma^{0}$ with $\psi^{\dagger}$ being the conjugate transpose of $\psi$. Spinor fields are naturally endowed with a norm $\left| \psi \right|^2_{\psi}$ induced by inner product $ \left\langle \psi_1 , \psi_2 \right\rangle := \bar{\psi}_1 \psi_2  \in \mathbb{C}$. Eq.~(\ref{SMat}) used the following definition for the gauge covariant exterior derivative $D_{\omega}: \Omega^0 \left( \mathcal{M} \right) \otimes \mathbb{G} \;\; \rightarrow \;\; \Omega^1 \left( \mathcal{M} \right) \otimes \mathbb{G}$:
\begin{align}\label{gaugecovder}
D_{\omega} \psi := \left(d + \frac{1}{4} \omega_{ab} \sigma^{ab} \right) \psi \quad &, \quad D_{\omega} \bar{\psi} := \overline{D_\omega \psi}
\end{align}
and we have written $\sigma^{a_1 a_2 \cdots a_k}:= \gamma^{\left[ a_1 \right.} \cdots \gamma^{\left. a_k \right]}$, i.e. the complete anti-symmetric product of the $\gamma^a$ matrices.

The total action (\ref{action1}) gives the following canonical field equations in the VEP formalism (see Appendix \ref{sec:cartan}) for the vielbein $e^a$ and connection $\omega^a_{\;\;b}$:
\begin{align}
\label{fe41} & \small{\delta e} : \; \Xi_a = - \left[ R^{\left( *\right)}_{ab} - \frac{2 \tilde{\lambda}}{3} \hat{\Sigma}^{(*)}_{ab} \right] \wedge e^{b} - i d \varphi_{N} \wedge T_{a} \;\;, \\
\nonumber & \small{\delta \omega} : \; \tau_{ab} = i 2 d \hat{\varphi}_P \wedge R_{ab} - i 2 d \hat{\varphi}_E \wedge R^{\left( * \right)}_{ab}  + \\ 
\label{fe42} & \quad \quad \quad \quad \quad \quad \quad \quad \quad  - i d \varphi_{N} \wedge \hat{\Sigma}^{\flat}_{ab} + d_{\omega} \hat{\Sigma}^{(*)}_{ab} \quad,
\end{align}
where we have defined the re-scaled quantities:
\begin{align}\label{rescaled}
& \hat{\varphi}_E := \frac{\varphi_E}{\left( 4 \pi \right)^2} \;\;\; , \;\;\; \hat{\varphi}_P := \frac{\varphi_P}{\left( 4 \pi \right)^2} \;\;\;, 
\end{align}
and we have written the $3$-form currents:
\begin{align}
\label{3currentsa} \frac{ \delta \mathfrak{L}_{M}}{\delta e^a} &:= \Xi_{a} = \frac{2}{3}\left[ 3 \mathfrak{Re} \left\{ \bar{\psi} \gamma^{b} D_{\omega} \psi \right\} + V e^b \right] \wedge \hat{\Sigma}^{(*)}_{ab} \;\; , \\
\label{3currentsb} \frac{\delta \mathfrak{L}_M}{\delta \omega^{ab} } &:= \tau_{ab} = \frac{1}{4} \bar{\psi} \sigma_{abc} \psi \star e^c \;\; .
\end{align}

Physically, it will be shown that the couplings zero forms $\varphi_j$ ($j \in J$) can be naturally interpreted as scalar matter fields. This becomes clear once the total $2$-form curvature $R_{ab}$ is obtained, as its torsion related part will have terms directly dependant on $\varphi_j$ with a clear dynamical interpretation. From another point of view, if we rightly define the cosmological functional $\tilde{\lambda}$, we can also understand the couplings $\varphi_j$ in the topological action terms as generalized Lagrange multipliers for the characteristic classes $C_j$. Pointing in that direction, in order for the variation with respect to the coupling zero-forms $\varphi_j$ to be well posed up to first order, the associated field equations must read:
\begin{align}
\label{fe43456} i \frac{ \delta \mathfrak{L}_G }{ \delta \varphi_{N}} = \small{ C_N } \;\; , \;\; i \frac{ \delta \mathfrak{L}_G} { \delta \varphi_E} =  \small{ C_E } \;\; , \;\; i \frac{ \delta \mathfrak{L}_G }{ \delta \varphi_P} = \small{ C_P } \;\; ,
\end{align}
However, from the definition of the gravitational action (\ref{EHaction}), particularized for the present case (\ref{EHactionmod}), we can only assume a functional dependency of these zero-forms in the \textit{cosmological functional} term. Since the spinor fields are not explicitly interacting, i.e. there is no interaction Lagrangian density defined above, we then consider the cosmological functional $\tilde{\lambda}$ as having this role of interacting potential. We expect the latter to have the form:
\begin{align}\label{lambdaint}
\tilde{\lambda} = \tilde{\lambda} \left[ \varphi_j , \Psi \right] := \bar{\psi} \, \lambda \left[ \varphi_j \right] \, \psi
\end{align}
with $\lambda \left[ \varphi_j \right] \in \Omega^0 \left( \mathcal{M} \right)$ analytical in $\varphi_j$. This is no other than a Yukawa type of interaction between the coupling zero-forms and the spinor fields. Hence, we can  equivalently write the set (\ref{fe43456}) as:
\begin{align}\label{lagrangemult}
\delta \varphi_j : \;\; - i \frac{\delta \lambda }{ \delta \varphi_j } \left| \psi \right|^2_{\psi} d \mu = C_j  \;\;\;\; , \;\;\;\; j \in J \;\; , 
\end{align}
where we have used spinor norm $\left| \psi \right|^2_{\psi} \in \mathbb{R}^+$, left open to be normalized at some point. When the latter is combined with Eq.~(\ref{Stop}) it allows us to write the current like terms as follows:
\begin{align*}
\int\limits_{\mathcal{M}} \frac{i}{\kappa} \left( \varphi_j C_j \right) & = \left| \psi \right|^2_{\psi} \left( \frac{\varphi_j}{\kappa} , \frac{\delta \lambda }{ \delta \varphi_j } \right) \,  | \, j \in J \; \text{(no summation)} \, ,
\end{align*}
where the round brackets are part of the $L^2$-inner product naturally induced by the Hodge dual operator in $\mathcal{M}$ (see its definition in Eq.~\ref{L2inner}). This is in complete analogy with the Schwinger formalism of current terms as sources of a field theory. In other words, it is explicit here that the gravitational part of the total action (\ref{action1}), through a mechanism in which the cosmological functional $\tilde{\Lambda}$ is central, becomes a source of the topology of the space-time manifold $\mathcal{M}$.

As discussed in Appendix \ref{TopManifold}, it is worth noticing that a direct consequence of the model presented is the following: the sole form of the set (\ref{lagrangemult}) ensures that the space-time manifold $\mathcal{M}$ is oriented, simply connected and smooth, all of which are completely reasonable topological requirements for any cosmological model. Moreover, by the above, the related set of field Eqs.~(\ref{fe43456}) ultimately result in well defined integer topological numbers $n_j$ for $j=E,P$ (see Appendix \ref{subsec:wellposed} for their definitions).

Finally, complementing the set of field Eqs.~(\ref{fe41}), (\ref{fe42}) and (\ref{lagrangemult}), when varying the action (\ref{action1}) with respect to the explicit matter fields we have:
\begin{align}
\nonumber \delta \bar{\psi} : \;\;\;\; & \left[ \gamma^a D_{\omega} \psi - \frac{\lambda}{4} \psi \, e^a  - \frac{1}{4} \frac{\delta V}{\delta \bar{\psi}} e^a \right] \wedge \star e^{\flat}_a = 0 \;\; , \\
\label{feD1} \delta \psi : \;\;\;\; & \left[ D_{\omega} \bar{\psi} \gamma^a + \frac{\lambda}{4} \bar{\psi} \, e^a + \frac{1}{4} \frac{\delta V}{\delta \psi} e^a \right] \wedge \star e^{\flat}_a = 0 \;\; , 
\end{align}
where we notice that we can interpret the functional $\lambda \left[ \varphi_j \right]$ as being responsible for giving mass to the spinor fields $\Psi$, and thus expected to be naturally greater than zero.

\subsection{\label{subsec:wellposed} Well posed-ness and topological consequences}

Considering a manifold with boundary, as well as the simultaneous upgrade of the quantities $\bar{\omega}^a_{\;\;b} \rightarrow \omega^a_{\;\;b}$, $\bar{R}^a_{\;\;b} \rightarrow R^a_{\;\;b}$ and $\Lambda \rightarrow \tilde{\lambda}$, in a gravitational action such as (\ref{action1}) only requires the addition of boundary counter terms, for the field equations to be well posed up to first order, with the standard boundary conditions: $\delta \omega_{ab}|_{\partial \mathcal{M}} = \delta e_a |_{\partial \mathcal{M}} = 0 $. This is equivalent, in standard GR, to the addition of the so called York-Gibbons-Hawkins terms at the level of the action. Since the other terms do not produce a surface term, we focus on the cancellation of the $\tilde{\lambda}$ term. The entire procedure results in the modification:
\begin{align}\label{EHmod}
\int\limits_{\mathcal{M}} {\Lambda} \, d \mu \quad \rightarrow \quad \int\limits_{\mathcal{M}} \tilde{\lambda} \, d \mu \,\, - \int\limits_{\partial \mathcal{M}} i_n \left( \tilde{\lambda} \, d \mu  \right)  \;\;.
\end{align}

Regarding now this choice, although Hadamard's conditions for a well posed problem have not been thoroughly studied in this case, some remarks about similar systems in the context of GR have been studied. For instance, a flat universe cosmological solution is known to be well posed \cite{Karp,Frittelli,Gundlach}. However, the systematic study of this question definitely falls out of the scope of the present work.

Now, for $\tilde{\lambda} \left[ \varphi_j , \Psi \right]:= \bar{\psi} \lambda \left[ \varphi_j \right] \psi $, as in Eq.~(\ref{lagrangemult}), the variation with respect to the coupling zero-forms $ \delta \varphi_j $ is uniquely satisfied by: 
\begin{align}\label{lagrangemult1}
\delta \varphi_j : \;\;\; C_j = - i \left| \psi \right|^2_{\psi} \frac{ \delta \lambda }{ \delta \varphi_j} d \mu \quad , \quad   \text{for} \quad  j \in J \;\; ,
\end{align}
which has topological consequences. The presented approach have led us to define a well posed functional $\lambda$ that encodes the topological information of the space-time manifold. Even further:

\begin{prop}\label{diagcoupling}
The manifold $\mathcal{M}$ of the gravitational theory described by the action (\ref{action1}) is characterized by integer topological numbers $n_j = \int\limits_{\mathcal{M}} C_j $, for $j = E , P$.
\end{prop}
\begin{proof}
From Eq.~(\ref{lagrangemult}) we have:
\begin{enumerate}[i]
\item De Rham's theorem allows us to interpret each $C_i$ as diagonal elements of the $H^2 \left( \mathcal{M} ; \mathbb{R} \right) \times H^2 \left( \mathcal{M} ; \mathbb{R} \right)$ kind, and thus, as diagonal intersection forms $Q_M$. 
\item Donaldson's theorem, regarding the $Q_M$ intersection forms, ensures that the manifold $\mathcal{M}$ is \textit{smooth and simply connected} \cite{donaldson,milnor}. 
\item We know that the manifold is \textit{oriented}, as testified from the existence of the well defined volume form $d \mu$ appearing in Eq.~(\ref{action1}).
\end{enumerate}
Therefore, we fall into the conditions of Remark \ref{topnumber}.
\end{proof}

In other words, epistemologically, we can turn this last proposition around and justify the appearance of the functional $\tilde{\lambda}$, as a way to ensure some \textit{minimum topological requirements} for our gravitational theory: \textit{smoothness, simple connected-ness and orientation}.

\section{\label{sec:sol}The solution}

In this section we focus on finding a consistent solution of the field Eqs.~(\ref{fe41} - \ref{fe43456}). We recall the following facts: since the space-time manifold $\mathcal{M}$ is assumed to be compact, but closed, the $0$-forms $\varphi_i$ coupled to the $C_i$ (see Eqs.~\ref{Pontryagin} - \ref{NiehYan}), are necessarily bounded.

Summarizing the previous discussions, to allow the manifold $\mathcal{M}$ to be compact, but not necessarily closed, requires the presence of boundary terms on the right hand side of Eq.~(\ref{lagrangemult1}). Such requirement acts as the so called York-Gibbons-Hawking terms in the case of the Einstein-Hilbert formalism \cite{Dyer}, and ensures the well posed-ness of the field equations.

We can justify these boundary terms:
\begin{enumerate}[a)]
    \item from an observational point of view, several astronomical physical objects are described as space-time singularities,
    \item from a theoretical perspective, one can isolate singular regions with closed hyper-surfaces on which the boundary conditions of physical observables are well behaved.
\end{enumerate}
The latter procedure effectively defines a non-trivial manifold boundary $\partial \mathcal{M}$ even if the volume of the manifold $\text{Vol} \left( \mathcal{M} \right)$ is taken to be infinite.

\subsection{\label{subsec:solution} An \textit{ad-hoc} contortion}

We begin by clarifying how GR could be recovered from the model presented in Eq.~(\ref{action1}). A null contortion is a sufficient condition to recover a GR-like behavior and absence of torsion, as Eq.~(\ref{contortion1}) defines a one-to-one correspondence between the torsion and contortion. More concretely, $T^a = 0$ iff $K^{a}_{\;\;b} =  0$. Hence, the connection $1$-form $\omega^a_{\;\;b}$ reduces to  the Levi-Civita $\bar{\omega}^a_{\;\;b}$. Such constraint result in an effective GR with some further topological restrictions. Hence, null contortion is at least a \textit{sufficient condition} for a GR-type theory.

From Eqs.~(\ref{fe41}) and (\ref{fe42}), a GR-like behavior in torsion-less regions of space-time occurs for slowly varying coupling zero-forms $\varphi_j$ ($j \in J$). We expect a limiting behavior of the type:
\begin{align}\label{limitcouplings}
\lim\limits_{ \left\| T \right\|_{L^2} \rightarrow 0} d {\varphi}_j := d {\varphi}_j |_0 = 0 \;\; ; \;\; j \in J \;\;,
\end{align}
for the coupling zero-forms, where we have written $\left\| T \right\|_{L^2} $ for the $L^2$-norm as defined in Eq.~(\ref{L2norm}). The absence of torsion is translated through Eq.~(\ref{limitcouplings}) into the vanishing of the coupling zero-forms exterior derivatives. The sufficiency of this restriction will be apparent below.

\subsection{\label{firsttwofe} Solving the VEP field equations.}

Using the decomposition (\ref{Riemanntensorgral}) we can write the field Eq.~(\ref{fe41}) in the equivalent form:
\begin{align*}
& \Xi_a = - \left\{ \bar{R}^{\left( * \right)}_{ab} - \frac{2 \tilde{\lambda}}{3} \hat{\Sigma}^{\left( * \right)}_{ab} + \Theta^{\left( * \right)}_{ab} + i d \varphi_{N} \wedge K_{ab} \right\} \wedge e^b 
\end{align*}
which allows us to keep track of the torsion related terms. Assuming the existence of solutions for the system of field equations (\ref{fe1} - \ref{fe4}), our main strategy will be to find a contortion depending on the free parameters of the system, mainly the coupling zero-forms $\varphi_j$ ($j \in J$ ), so it can be inserted back in the field equations and later particularized. We thus consider the class: 
\begin{align*}
K_{ab} \sim \star \left( d \varphi_N \wedge \left[ \alpha \, \hat{\Sigma}^{\flat}_{ab} + \beta \, \hat{\Sigma}^{(*)}_{ab} \right] \right) \;\, , \,\; \alpha, \beta \in \mathbb{C}
\end{align*}
which profiles as a family of possible candidates for the following reasons: a) it can be shown to be consistent with Eq.~(\ref{fe41}) and b) it does not break the symmetries of the field equations, especially the Einstein-type (See Definition \ref{Einstein}). Although we have not proved the non existence of a solution for a general combination $0 \neq \alpha , \beta \in \mathbb{C}$, it seems that considering both parameters introduce non-linear terms by means of calculating the torsion related part of the curvature $2$-form $\Theta_{ab}$ that compromise the possibility of obtaining the right torsion-less behavior for Eq.~(\ref{fe42}). Furthermore, a simple heuristic condition such as Eq.~(\ref{limitcouplings}) for the coupling zero-forms profiles is not sufficient to counter balance these non-linearities. However, the subclass of contortions with $\beta = 0$ does. In fact, we can conveniently choose $\alpha = 4 i$ to write explicitly:
\begin{align}\label{contortionfin} 
K^a_{\;\;b} = - \frac{i}{3} \epsilon^{a \;\; c}_{\, \cdot \, b \, \cdot \, d} \mathcal{L}_c  \left( \varphi_N \right) e^d \;\; \in \;\; \Omega^1 \left( \mathcal{M} \right) \;\;\; ,
\end{align}
where the torsion can be easily obtained by using Eq.~(\ref{contortion})  to yield:
\begin{align}\label{torsionfin}
T^a = \frac{i}{3} \star \left( d \varphi_N \wedge e^a \right) \;\;\; \in \;\; \Omega^2 \left( \mathcal{M} \right) \;\;.
\end{align}

\begin{rem}\label{torsioncond}
From the last expression, the restriction expressed in (\ref{limitcouplings}) is a necessary and sufficient condition for the null torsion.
\end{rem}

Continuing, by means of the following result:
\begin{align}\label{minires}
d_{\omega} \mathcal{L}_a \left( \varphi_N \right) = d \mathcal{L}_a \left( \varphi_N \right) = - \Delta \varphi_N e^{\flat}_a \;\;,
\end{align} 
as well as with Eqs.~(\ref{torcurv}) and (\ref{minires}), we can calculate the torsion related part of the curvature $2$-form and its Lie dual to yield:
\begin{align}
\nonumber & \Theta_{ab} = i d \varphi_N  \wedge \frac{\epsilon_{ab}^{\;\;\;\;cd} K_{cd}}{2} - \frac{4 i \Delta \varphi_N}{3} \hat{\Sigma}^{(*)}_{ab} - \frac{2 \left| d \varphi_N \right|^2 }{3}  \hat{\Sigma}^{\flat}_{ab} \\
\label{DeltaR} & \quad \quad \quad  \text{and}
\\
\nonumber & \Theta^{(*)}_{ab} = - i d \varphi_N \wedge K_{ab} - \frac{4i \Delta \varphi_N}{3} \hat{\Sigma}^{\flat}_{ab} - \frac{2 \left| d \varphi_N \right|^2 }{3} \hat{\Sigma}^{(*)}_{ab} \quad, 
\end{align}
where we are writing $\left\langle d \varphi_N , d \varphi_N \right\rangle := \left| d \varphi_N \right|^2$ for short. This is a kinetic term for the coupling zero-form, effectively turning this elements into a \textit{matter scalar field}. Thus, once inserting the second equation from the set (\ref{DeltaR}) back in Eq.~(\ref{fe41}) we obtain the expression:
\begin{align}\label{fieldequation1} 
\Xi_a & = - \left\{ \bar{R}^{\left( * \right)}_{ab} -  \frac{ 2 \left( \tilde{\lambda} + \left| d \varphi_N \right|^2 \right) }{3} \hat{\Sigma}^{(*)}_{ab} \right\} \wedge e^b \;,
\end{align}
which is formally equivalent to the GR field equation. Here, the term proportional to the contortion has been appropriately canceled, as wanted.

We can further simplify this equation by means of Eqs.~(\ref{3currentsa}) and (\ref{feD1}), to write:
\begin{align}
\label{mattercurrent} \Xi_a & = \frac{2}{3} \xi \left[ \Psi \right] \hat{\Sigma}^{(*)}_{ab} \wedge e^b \quad , \quad \text{where} \\
\nonumber \xi \left[ \Psi \right] & := \frac{3}{4} \tilde{\lambda} + \frac{3}{8} \left\{ \bar{\psi}\frac{\delta V \left[ \Psi \right]}{\delta \bar{\psi}} + \frac{\delta V \left[ \Psi \right]}{\delta \psi} \psi \right\} +  V \left[ \Psi \right] \;\;.
\end{align}
So we can equivalently write (\ref{fieldequation1}) as:
\begin{align*} 
0 & = \left[ \bar{R}^{\left( * \right)}_{ab} -  \frac{ 2 \left( \tilde{\lambda} + \left| d \varphi_N \right|^2 - \xi \right)}{3} \hat{\Sigma}^{(*)}_{ab} \right] \wedge e^b \;\; ,
\end{align*}
where we have noted $\xi := \xi \left[ \Psi \right]$ for simplicity. It can be inverted by means of the Lemma (\ref{curvA}) and the slant product to obtain the preliminary results for the torsion-less curvature $\bar{R}_{ab}$:
\begin{align}\label{RbarApp1}
\bar{R}_{ab} & = \frac{ 2 \left( \tilde{\lambda} + \left| d \varphi_N \right|^2 - \xi \right) }{3} \hat{\Sigma}^{\flat}_{ab} \quad , 
\end{align}
or, by means of Eq.~(\ref{mattercurrent}):
\begin{align}
\nonumber & \bar{R}_{ab} = \frac{ 2 \left( \alpha_{\psi} \lambda \left[ \varphi_N \right] + \left| d \varphi_N \right|^2 - U \left[ \Psi \right] \right) }{3} \hat{\Sigma}^{\flat}_{ab} \;\; , \\
\nonumber & \quad \text{where} \quad \quad \quad \alpha_{\psi} := \frac{\left| \psi \right|^2_{\psi}}{4} \quad \text{and} \\ 
\label{RbarApp2} & U \left[ \Psi \right] := \frac{3}{8} \left\{ \bar{\psi}\frac{\delta V \left[ \Psi \right]}{\delta \bar{\psi}} + \frac{\delta V \left[ \Psi \right]}{\delta \psi} \psi \right\} +  V \left[ \Psi \right] \;\; .
\end{align}
Moving forward, by means of Eqs.~(\ref{Riemanntensorgral}), (\ref{DeltaR}) and (\ref{RbarApp2}) in its first form, we can write the total curvature $2$-form as:
\begin{align*}
R_{ab} = \frac{2 \left( \alpha_{\psi} \lambda - U \right)}{3} \hat{\Sigma}^{\flat}_{ab} - \frac{4 i \Delta \varphi_N}{3} \hat{\Sigma}^{(*)}_{ab} + i d \varphi_N  \wedge \frac{\epsilon_{ab}^{\;\;\;\;cd} K_{cd}}{2}
\end{align*}

For definiteness, we take $V \left[ \Psi \right] = 0$ or equivalently $U \left[ \Psi \right] = 0$, since this is the only term that explicitly depends on the spinor fields. Thus the final form for the curvature $2$-form and its Lie dual are: 
\begin{align}
\nonumber R_{ab} & = \frac{2 \alpha_{\psi} \lambda }{3} \hat{\Sigma}^{\flat}_{ab} - \frac{4 i \Delta \varphi_N}{3} \hat{\Sigma}^{(*)}_{ab} + i d \varphi_N  \wedge \frac{\epsilon_{ab}^{\;\;\;\;cd} K_{cd}}{2} \\
\label{curvaturetotal} & \quad \quad \text{and} \\
\nonumber R^{(*)}_{ab} &  = \frac{2 \alpha_{\psi} \lambda}{3} \hat{\Sigma}^{(*)}_{ab} + \frac{4 i \Delta \varphi_N}{3} \hat{\Sigma}^{\flat}_{ab} - i d \varphi_N  \wedge K_{ab} \;\; .
\end{align}

Let us now tackle the field Eq.~(\ref{fe42}), which by means of Eq.~(\ref{curvaturetotal}) can be written: 
\begin{align}
\nonumber & \tau_{ab} = d_{\omega} \hat{\Sigma}^{(*)}_{ab} - \frac{4i}{3} \left[ \alpha_{\psi} \lambda d \hat{\varphi}_E + \left( i 2 \Delta \varphi_N \right) d \hat{\varphi}_P \right] \wedge \hat{\Sigma}^{(*)}_{ab} + \\ 
\nonumber  & + \frac{4 i}{3} \left[ \alpha_{\psi} \lambda d \hat{\varphi}_P - \left( i 2 \Delta \varphi_N \right) d \hat{\varphi}_E - \frac{3}{4} d \varphi_{N} \right] \wedge \hat{\Sigma}^{\flat}_{ab} + \\
\label{jacobians0} & + 2 d \varphi_N  \wedge \left[ \left( d \hat{\varphi}_E \wedge K_{ab} \right) + \left( d \hat{\varphi}_P \wedge K_{ab} \right)^{(*)} \right] \;\; .
\end{align}

When taking the following natural quadratures:
\begin{align}
\label{jacobians1} d \hat{\varphi}_E & = \frac{3}{4} \frac{ i 2 \Delta \varphi_N }{ \left( \alpha_{\psi} \lambda \right)^2 + \left( 2 \Delta \varphi_N \right)^2} d \varphi_N \;\; , \\
\nonumber d \hat{\varphi}_P & = \frac{3}{4} \frac{ \alpha_{\psi} \lambda }{ \left( \alpha_{\psi} \lambda \right)^2 + \left( 2 \Delta \varphi_N \right)^2 } d \varphi_N \;\; ,
\end{align}
we are left with the equation:
\begin{align}\label{fieldequation2}
\tau_{ab} = d_{\omega} \hat{\Sigma}^{(*)}_{ab} \;\; .    
\end{align}
which is formally equivalent to the analogous GR equation with torsion.

Notice that the set (\ref{jacobians1}) suggests that we can consider the zero-forms $\varphi_E$ and $\varphi_P$ as being:
\begin{align}
\label{couplingrestriction} & d \varphi_j := {\varphi_N}^* \, d \phi_j = d \left( \phi_j \circ \varphi_N \right) \;\; \text{for} \;\; j = E, P 
\end{align} 
where we have denoted the pullback of $\varphi_N$ by $ {\varphi_N}^*$ and we expect $\phi_j \in \Omega^0 \left( \mathcal{M} \right)$ to be a common diffeomorphism,  with $j = E,P$.

\subsection{\label{firsttwofe} The topological field equations and the appearance of the cosmological constant}

Since we have obtained a torsion $2$-form and a total curvature $2$-form, the previous section put us in a position to directly calculate the characteristic classes. This will allow us to solve the rest of the field equations. Let us begin by calculating the Nieh-Yan characteristic class $C_N$. From Eqs.~(\ref{NiehYan}) and Eq.~(\ref{contortionfin}), we obtain:
\begin{align}\label{NYdirect}
C_N = - 2 i d \star d \varphi_N = - i \left( 2 \Delta \varphi_N \right) d \mu \;\; ,
\end{align}
thus, combining the previous with Eq.~(\ref{lagrangemult}) for $j=N$,  we get:
\begin{align}\label{deltalambda}
\frac{\left| \psi \right|^2_{\psi}}{2} \frac{\delta \lambda }{ \delta \varphi_N } = \Delta \varphi_N \Rightarrow \lambda = - \frac{\left| d \varphi_N \right|^2}{\left| \psi \right|^2_{\psi}} + \frac{4 \Lambda}{\left| \psi \right|^2_{\psi}},
\end{align}
where the last expression is obtained when considering the restriction (\ref{limitcouplings}) on the torsion-less domain and we identify $\Lambda$ as being the \textit{topological cosmological constant} 

\begin{rem}\label{Rbarfin}
We remark that expression (\ref{deltalambda}) is completely consistent with Eq.~(\ref{RbarApp1}). In fact when inserting the latter we obtain:
\begin{align*}
\bar{R}_{ab} = \frac{2}{3} \Lambda \hat{\Sigma}^{\flat}_{ab} - \frac{2}{3} U \left[ \Psi \right] \hat{\Sigma}^{\flat}_{ab} \quad , 
\end{align*}
which is analogous to GR with a cosmological constant %like 
term and a matter component $2$-form: $-\frac{2}{3} \xi \left[ \Psi \right] \hat{\Sigma}^{\flat}_{ab}$.
\end{rem}

Similarly, we can directly calculate the Euler characteristic $C_E$ using Eq.~(\ref{curvaturetotal}) and expression (\ref{Euler}) to obtain:
\begin{align}
\nonumber \left( 4 \pi \right)^2 C_E & = - R^{ab} \wedge R^{\left( * \right)}_{ab} \\
\nonumber & = - \{\frac{2 \alpha_{\psi} }{3} \lambda  \hat{\Sigma}^{ab} + \frac{4}{3} \left( i \Delta \varphi_N \right) \star \hat{\Sigma}^{ab} + \\
\nonumber & - i d \varphi_N  \wedge \frac{1}{2} \epsilon^{ab}_{\;\;\,cd} K^{cd} \} \wedge \{ \frac{2 \alpha_{\psi} }{3} \lambda  \hat{\Sigma}^{(*)}_{ab}  + \\
\nonumber & - \frac{4}{3} \left( i \Delta \varphi_N \right) \hat{\Sigma}^{\flat}_{ab} - i d \varphi_N  \wedge K_{ab} \} \\
\label{CE12} & = - \frac{4}{3} \left[ \left( \alpha_{\psi} \lambda \right)^2 + \left( 2 \Delta \varphi_N \right)^2 \right] d \mu \quad .
\end{align}
When combining the previous equation with Eq.~(\ref{lagrangemult}) for $j=E$ and Eq.~(\ref{couplingrestriction}), we get:
\begin{align*}
\frac{4}{3} \left[ \left( \alpha_{\psi} \lambda \right)^2 + \left( 2 \Delta \varphi_N \right)^2 \right] = i \left\{ \frac{ \partial \hat{\phi}_E }{ \partial \varphi_N} \right\}^{-1} \frac{ \delta \lambda }{ \delta \varphi_N} \;\; ,
\end{align*}
which is completely compatible if the same quantity is calculated by means of Eqs.~(\ref{deltalambda}) and (\ref{jacobians1}).

The Pontryagin characteristic class can also be calculated by means of Eqs.~(\ref{Pontryagin}) and (\ref{curvaturetotal}), to yield:
\begin{align}
\nonumber \left( 4 \pi \right)^2 C_P & = - 2 R^{ab} \wedge R_{ab}  \\
\nonumber & = - 2 \{\frac{2 \alpha_{\psi} }{3} \lambda  \hat{\Sigma}^{ab} + \frac{4}{3} \left( i \Delta \varphi_N \right) \star \hat{\Sigma}^{ab} + \\
\nonumber & \quad - i d \varphi_N  \wedge \frac{1}{2} \epsilon^{ab}_{\;\;\,cd} K^{cd} \} \wedge \{ \frac{2 \alpha_{\psi} }{3} \lambda \hat{\Sigma}^{\flat}_{ab}  + \\
\nonumber & \quad + \frac{4}{3} \left( i \Delta \varphi_N \right) \hat{\Sigma}^{(*)}_{ab} + \\ 
\nonumber & - i d \varphi_N  \wedge \frac{\epsilon_{ab}^{\;\;\;\; cd} K_{cd}}{2} \} \\
\label{CP12} & = - \frac{ 16 i }{3} \left(  \alpha_{\psi} \lambda \right) \left( 2 \Delta \varphi_N \right) d \mu \quad. 
\end{align}
thus, when combining the previous equation with Eq.~(\ref{lagrangemult}) for $j=P$ and Eq.~(\ref{couplingrestriction}), we get:
\begin{align*}
\frac{16}{3} \left( \alpha_{\psi} \lambda \right) \left( 2 \Delta \varphi_N \right) = \left\{ \frac{ \partial \hat{\phi}_P }{ \partial \varphi_N} \right\}^{-1} \frac{ \delta \lambda }{ \delta \varphi_N} \;\; .
\end{align*}
When substituting Eqs.~(\ref{deltalambda}) and (\ref{jacobians1}) we obtain the restriction:
\begin{align}\label{Laplacian1}
3 \left( \alpha_{\psi} \lambda \right)^2 = \left( 2 \Delta \varphi_N \right)^2 \; \Rightarrow \; \Delta \varphi_N = \pm \frac{\sqrt{3}}{2} \alpha_{\psi} \lambda \;\;.
\end{align}
When taking Eqs.~(\ref{deltalambda}) and (\ref{RbarApp2}) into account, we obtain:
\begin{align}\label{deltalambda1}
\frac{\delta \ln \left| \lambda \right| }{\delta \varphi_N} = \pm \frac{\sqrt{3}}{4} \, \Rightarrow \, \lambda = \frac{4 \Lambda}{\left| \psi \right|^2_{\psi}} \exp \left( \pm \frac{\sqrt{3}}{4} \varphi_N \right) \; ,  
\end{align}
where we have chosen the amplitude  and the sign of the argument in the exponential function concomitantly with that of the relation (\ref{deltalambda}). In other words, the following expression is then satisfied:
\begin{align}\label{varphi}
\left| \frac{d \varphi_N}{4} \right|^2 & = \Lambda \left\{ 1 - \exp \left( \pm \frac{\sqrt{3}}{4} \varphi_N \right) \right\} \;\; ,
\end{align}
and thus can be used to characterize the null torsion domain (i.e. the GR regions) in a sufficient way: $d \varphi_N = 0 \Rightarrow \left| d \varphi_N \right|^2 = 0 \Leftrightarrow \varphi_N = 0$. It follows from there that the Eikonal equation (second case)\footnote{ This equation models the exact propagation of electromagnetic wave fronts, independently of the wave lengths or wave structure (a result that follows from the theory of partial differential equations of second
order in characteristic manifolds) \cite{bateman,Fritelli,Newman}.} characterizes the GR regions coinciding with $\varphi_N = 0$. More formally, we can define the set: 
\begin{align}\label{Omega1}
\mathcal{M} \supseteq \mathcal{N} := \left\{ x \in \mathcal{M} \; | \; d \varphi_N \left( x \right) = 0 \right\} \;\;\; ,
\end{align}
which by means of restriction (\ref{limitcouplings}) and remark (\ref{torsioncond}), is equivalent to the \textit{domain with compact support in which torsion is null}. Basically, this is the entire region of space-time manifold $\mathcal{M}$ where gravity is equivalent to GR.

\subsection{\label{BVproblem} The boundary value problem and stability.}

To complete the picture, let us write the rescaled quantity $\hat{\varphi}_N := \frac{\sqrt{3}}{4} \varphi_N$ and make the reasonable assumptions that $\mathcal{N}$, as defined in Eq.~(\ref{Omega1}) is simply connected as well and has a smooth boundary $\partial \mathcal{N} \neq \emptyset$. It then follows that:
\begin{align}\label{boundlambda}
\lambda \left[ \varphi_N \right] = \lambda \left[ 0 \right] = \Lambda \quad , \quad \forall \;\; x \in {\mathcal{N}} \;\;.
\end{align}
We can then write:
\begin{align}\label{manifoldregions}
\mathcal{M} = \mathcal{N} \sqcup \mathcal{M} \backslash \mathcal{N}  
\end{align}
and assume that $\partial \left( \mathcal{M} \backslash \mathcal{N} \right) = \emptyset $.

At the same time, Eq.~(\ref{varphi}) does not offer a physical understanding of the stable behavior of that functional, of the kind defined in Eq.~(\ref{boundlambda}). Here we can interpret the kinetic term (\ref{varphi}) as a Lyapunov-type of potential that creates a strong attractor for the cosmological function $\lambda$. With the modification:
\begin{align}\label{varphifinal}
\left| \frac{d \varphi_N}{4} \right|^2 & = \Lambda \left\{ 1 - \exp \left( - \frac{\sqrt{3}}{4} \left| \varphi_N \right| \right) \right\} \;\; ,
\end{align}
this potential exactly satisfies the requirements (\ref{Liapunov}) of a Lyapunov one, where the domain $\mathcal{N}$, defined in Eq.~(\ref{Omega1}), is in this case the domain of attraction for the \textit{dynamical system} defined by the field equations (\ref{fe41} - \ref{fe43456}) over the entire manifold $\mathcal{M}$.

The cosmological functional characterization is then equivalent to a boundary value problem satisfied by the \textit{scalar matter field} $\hat{\varphi}_N$ over the different regions of the $\mathcal{M}$, decomposed as in Eq.~(\ref{manifoldregions}):
\begin{prop}[Boundary value problem]\label{boundaryvalue}
The action (\ref{action1}) with restrictions (\ref{couplingrestriction}) defines \begin{itemize}[-]
    \item the boundary value problem over the domains $\mathcal{N}$ with boundary $\partial \mathcal{N}$:
\begin{align*}
\begin{cases}
\Delta \hat{\varphi}_N = \frac{ 3 \Lambda }{16} &, \;\; d^{\star} \hat{\varphi}_N \\
i^*_{\mathcal{N}} \hat{\varphi}_N = 0  &
\end{cases}
\end{align*}
where $i^*_{\mathcal{N}}: \Omega^k \left( \mathcal{N} \right) \rightarrow \Omega^k \left( \partial \mathcal{N} \right)$, and 
    \item the boundary value problem over the domain $\mathcal{M}^c = \mathcal{M} \backslash \mathcal{N}^o$ with boundary $\partial \mathcal{M}^c = \partial \mathcal{N} \sqcup \partial \mathcal{M}$:
\begin{align*}
\begin{cases}
\Delta \hat{\varphi}_N = \frac{ 3 \Lambda }{16} u_{\hat{\varphi}_N} &, \;\; d^{\star} \hat{\varphi}_N \\
i^*_{\mathcal{M}^c} \hat{\varphi}_N = 0  &, \;\; \text{with} \;\; u_{\hat{\varphi}_N} := \exp \left( - \left| \hat{\varphi}_N \right| \right)
\end{cases}
\end{align*}
where $i^*_{\mathcal{M}^c}: \Omega^k \left( \mathcal{M}^c \right) \rightarrow \Omega^k \left( \partial \mathcal{N} \right)$.
\end{itemize} 
\end{prop}

Be that as it may, a boundary value problem such as that of proposition~(\ref{boundaryvalue}) have an associated boundary value problem for some harmonic form that remains to be defined, in which case techniques such as the ones discussed in \cite{BELISHEV2008128} can be used. However, this is out of the scope of the present paper. Instead, we will try to obtain reasonable estimates that will help in connecting these quantities with a cosmological constant measurement.

\subsection{\label{subsec:characteristicclasses} Estimation of the cosmological functional.}

The rest of the section will follow the notation from Eq.~(\ref{Omega1}) onward, $u_{\hat{\varphi}_N} := \exp \left( - \left| \hat{\varphi} \right| \right)$, with $\hat{\varphi}_N = \frac{\sqrt{3}}{4} \varphi_N$ to estimate the functional $\lambda \left[ \varphi_N \right]$, we start by noticing that the Ricci $\left( Ric := R \right)$ scalar satisfies the following expression:
\begin{align}\label{Ricciscalar1}
\frac{R}{2} = \star \left[ R^{(*)}_{ab} \wedge \hat{\Sigma}^{ab} \right] = 2 \lambda \left| \psi \right|^2_{\psi} \, \Leftrightarrow \, R = 4 \lambda \left| \psi \right|^2_{\psi} \, .
\end{align}
Since functions in a compact manifold are bounded, in particular $\lambda$ and $\left| \psi \right|^2_{\psi}$, we can write:
\begin{align}\label{bounded}
R \in 16 \Lambda \cdot \left[ \min u_{\hat{\varphi}_N} , \max \,  u_{\hat{\varphi}_N} \right] \;\; .
\end{align}
The Ricci scalar $R$ is therefore bounded as well, which will prove to be of importance briefly.

We continue by calculating the topological numbers $n_i$, from Eq.~(\ref{topnumbers}).

\begin{enumerate}[i)]
\item ($n_N$) From Eq.~(\ref{NYdirect}) we have:
\begin{align} 
\label{nY11} n_N & = \mathfrak{Re} \left\{ - 2 i \int\limits_{\mathcal{M}} \Delta \varphi_N d \mu \right\} = 0 \quad .
\end{align}

We can consider this result as a somewhat partial answer to the question of why there is no global manifestation of a physical quantity that hints for the presence of torsion in a space-time manifold. 

\item ($n_P$) From Eq.~(\ref{CP12}) we have:
\begin{align}
\label{nP11} n_P & = \mathfrak{Re} \left\{ - \frac{16 i}{ \sqrt{3} \left(4 \pi \right)^2 } \int\limits_{\mathcal{M}} \left( \alpha_{\psi} \lambda\right)^2 d \mu \right\} = 0 \;\;.
\end{align}

This result can be combined with the Hirzebruch signature theorem to return:
\begin{align}
\nonumber & n_P = 3 \left( b^+ - b^- \right) = 0 \;\; , \\
\label{nPtop} \text{or} \;\;,\;\;\;\;\; & b^+= b^-:= b \;\;\;\;.
\end{align}
where $b^+$, $b^-$ are the dimensions of maximal positive $\mathcal{H}^+$ and negative $\mathcal{H}^-$ subspaces for the form in $H^2 \left( \mathcal{M} ; \mathbb{R} \right) = \mathcal{H}^+ \oplus \mathcal{H}^-$, respectively.  

\item ($n_E$) Finally, from Eq.~(\ref{CE12}) we have
\begin{align}
\nonumber n_E & = - \frac{16}{ 3 \left( 4 \pi^2 \right)} \mathfrak{Re} \left\{ \left\| \left( \alpha_{\psi} \lambda \right)^2 \right\|^2_{L^2} \right\} \\ 
\label{nE11} & = - \frac{16 \Lambda^2}{3 \left( 4 \pi \right)^2} \left\| u^2_{\hat{\varphi}_N} \right\|^2_{L^2} \;\; ,
\end{align}
independent of the norm $\left| \psi \right|^2$.
\end{enumerate}

We remark that this last quantity should be finite by topological arguments, which includes the squared norm of $u^2_{\hat{\varphi}_N}$. Furthermore, $n_E$ can be written in terms of the Betti numbers $b_j := b_j \left( \mathcal{M} \right) \geq 0$ for $j=0, ... , 4$ defined in Eq.~(\ref{betti}), as $n_E := \sum_j \left( -1 \right)^j b_j$. At the same time, by Poincair\'{e} duality  we have that $b_j = b_{4-j}$ for $j = 0, ... , 4$. Additionally, the second Betti number $b_2$ can be expressed as $b_2 := b^+ + b^- = 2 b$ by means of Eq.~(\ref{nPtop}), and the fact that $\mathcal{M}$ is thought to be simply connected, and thus $b_0 = 1$. Altogether, we can write $n_E$ as:
\begin{align}
\label{nEtop} n_E & = 2 + 2 b - 2 b_3 = - 2 k^2_E b_3.
\end{align}
where we have defined:
\begin{align}\label{kE}
k^2_E := 1 - \frac{1 + b}{ b_3}  \lesssim 1 \quad .
\end{align}
The last estimation comes from the claim that, $b_3$ should dominate the previous sum (\ref{nEtop}) for any sensible space-time manifold $\mathcal{M}$, hence its factoring out. This is based on the geometrical interpretation of $b_3$ as the \textit{number of $3$-punctures in $\mathcal{M}$}, which allows us to basically associate this number with the approximate number of GR singularities, i.e black holes in the universe. We notice that these Betti numbers are not in contradiction with a result for non-positive Ricci curvature manifolds \cite{wu1988bochner, 10.2307/1969287}. Among other things this result was obtained for closed manifolds, i.e. without boundary.

From the previous discussion, it follows that the exact expression for $\Lambda$ can be obtained as:
\begin{align}\label{lambdaexact}
\Lambda^2 = \frac{ 3 \left( 4 \pi \right)^2 k^2_E b_3 }{8 \left\| u^2_{\hat{\varphi}_N} \right\|^2_{L^2}} \quad ,
\end{align}
for which we will give some estimates. We begin by citing Eq.~(\ref{Omega1}) again, it immediately follows that we can write the following decomposition $\mathcal{M} = \mathcal{N} \,\, \sqcup \,\, \mathcal{M} \backslash \mathcal{N}$, thus:
\begin{align}
\nonumber 0 \leq \left\| u^2_{\hat{\varphi}_N} \right\|^2_{L^2} & = \text{Vol} \left( \mathcal{N} \right) + \left\| u^2_{\hat{\varphi}_N} \right\|^2_{L^2 \left( \mathcal{M} \backslash \mathcal{N} \right)} \\
\label{uestimate} \simeq \text{Vol} \left( \mathcal{N} \right) & = \frac{\text{Vol} \left( \partial \mathcal{N} \right)}{\mathcal{C}^2} < \infty
\end{align}
where we have used the result (\ref{boundlambda}) and the hypothesis of the strong attractor for the domain $\mathcal{N}$ and consequently a rapid decay for the function $u_{\hat{\varphi}_N}$ in the domain $\mathcal{M} \backslash \mathcal{N}$. Thus the volume $\text{Vol} \left( \mathcal{N} \right)$ dominates over the complete result. The last expression is obtained  by considering $\mathcal{N}$ to be a $4$-dimensional Cheeger minimizer so that $\mathcal{C}^2$ is the Cheeger or isoperimetric constant \cite{Benson}, having dimensions of $\left[ \texttt{length} \right]$. A less stringent result can also be obtained via the isoperimetric inequality, valid in general for manifolds with bounded Ricci curvature \cite{ASENS,hanes1972,Croke1984}, as it is argued to be the case in Eq.~(\ref{bounded}). The Cheeger constant is usually not known for general manifolds, but it is proven to be determined by the dimension of the manifold among other fixed topological parameters.

%Considering now that the domain $\mathcal{N}$ is compact, the closed-ness and bounded-ness of the region allows us to define a second constant $\mathcal{C} \leq \mathcal{D} \in \mathbb{R}$ with the same dimensions as $\mathcal{C}$ such the bracketing:
%\begin{align}\label{bracket1}
%\frac{1}{\mathcal{D}^2}\, \frac{\text{Vol} \left( \partial \mathcal{N} \right)}{4} \leq \left\| u^2_{\hat{\varphi}_N} \right\|^2_{L^2} \leq \frac{1}{\mathcal{C}^2} \, \frac{\text{Vol} \left( \partial \mathcal{N} \right)}{4} \;\; ,
%\end{align}
%is satisfied. Which, when combined with (\ref{lambdaexact}) gives the final bounded expression:
%\begin{align}\label{eq:Topolambda}
%\frac{ 2 \pi k_E \mathcal{C} }{ \left\langle \frac{2}{3} \text{Vol} \left( \partial \mathcal{N} \right) \right\rangle^{\frac{1}{2}} } \lesssim \Lambda \lesssim  \frac{ 2 \pi k_E \mathcal{D}}{ \left\langle \frac{2}{3} \text{Vol} \left( \partial \mathcal{N} \right) \right\rangle^{\frac{1}{2}} } \;\; ,
%\end{align}
Eq.~(\ref{uestimate}), combined with Eq.~(\ref{lambdaexact}) gives:
\begin{align}\label{eq:Topolambda1}
 \Lambda \approx  \frac{ 2 \pi \, k_E \, \mathcal{C} }{ \left\langle \frac{2}{3} \text{Vol} \left( \partial \mathcal{N} \right) \right\rangle^{\frac{1}{2}} }  \;\; ,
\end{align}
where we have written:
\begin{align*}
\left\langle  \text{Vol} \left( \partial \mathcal{N} \right) \right\rangle := \frac{1}{b_3}  \text{Vol} \left( \partial \mathcal{N} \right) \;\; ,
\end{align*}
as it was equivalently defined in Eq.~(\ref{eq:AvgBounVol}). Notice that, by Eqs.~(\ref{manifoldregions}) and (\ref{eq:Topolambda1}) we have the final result:
\begin{align}\label{eq:Topolambda}
 \Lambda \approx  \frac{ 2 \pi \, k_E \, \mathcal{C} }{ \left\langle \frac{2}{3} \text{Vol} \left( \partial \mathcal{M} \right) \right\rangle^{\frac{1}{2}} }  \;\; .
\end{align}

\section{\label{sec:borel} On Borel sets and Lyapunov stability of dynamical systems for manifolds}

This section consist of a very brief and incomplete presentation of the core material complementing section \ref{BVproblem} in order to support the interpretation of the scalar matter term as a dynamical stabilizer. More concretely, in the case of this paper we want to interpret the kinetic term of the scalar matter field, i.e. the term proportional to $\left| d \varphi_N \right|^2$ as a Lyapunov function that becomes an attractor for the cosmological functional to drive its value towards a constant.

We begin by defining what we mean by a dynamical system in this context. The latter is governed by differential equations on a manifold $\mathcal{M}$ that have the local representation \cite{Taringoo,Sontag}
\begin{align}
\nonumber \dot{x} (t) & = f (x (t), t), \\
\label{dynsystem} x (t_0) & = x_0 \in \mathcal{M}, \, t \in \left[t_0 , t_f \right].
\end{align}
where $f$ belongs to the space of smooth time varying vector fields $\mathfrak{X} \left( \mathcal{M} \times \mathbb{R} \right)$. The flow associated with $f$ is a map $\Phi_f$ satisfying
\begin{align*}
\nonumber \Phi_f : \left[ t_0 , t_f \right] \times \left[t_0 , t_f \right] \times \mathcal{M} \rightarrow \mathcal{M}  \\
\nonumber \left(s_0 , s_f , x \right) \mapsto \Phi_f \left( s_f , s_0 , x \right) \in \mathcal{M}
\end{align*}
and
\begin{align*}
\frac{d \Phi_f \left(s, s_0 , x \right)}{ds} |_{s=t} = f \left( \Phi_f (t, s_0 , x), t \right).
\end{align*}

\begin{defn}\label{Liapunov}
Suppose that $\mathcal{S}$ is an invariant set of the dynamical system $\left( \mathcal{M} , \Phi_f \right)$, where $\Phi_f : \left[ 0 , + \infty \right) \times \mathcal{M} \rightarrow \mathcal{M} $ is a continue semi-flow. Let $\mathcal{A}$ the domain of attraction defined by the semi-flow $f$. A continuous function $V : \rightarrow \mathbb{R}_{\geq 0}$ is a Lyapunov function if:
\begin{enumerate}[i)]
\item $V \left( x \right) > 0$ for all $x \in \mathcal{A} \backslash \mathcal{S}$,
\item $V \left( x \right) = 0$ for all $x \in \mathcal{S}$,
\item $V$ is proper, that is, $V^{-1} \left( B \right)$ is compact for every compact subset $B$ of
$\mathbb{R}_{\geq 0}$,
\item $V$ satisfies the condition $L_{f} V \left( x \right) = d V \left( f \right) < 0$ for all $x \in \mathcal{A} \backslash \mathcal{S}$ and $f: \mathcal{M} \rightarrow \mathcal{M}$
\end{enumerate}
\end{defn}

Let us suppose the existence of $\bar{x}$, a uniformly Lyapunov stable point. This means that, for any neighborhood $\mathcal{U}_{\bar{x}}$ of $\bar{x} \in  \mathcal{M}$ and any initial parameter $t_0 \in \mathbb{R}$, there exists a neighborhood $\mathcal{W}_{\bar{x}}$ of $\bar{x}$, such that $\forall \, x_0 \in \mathcal{W}_{\bar{x}}, \, \Phi_f \left( t , t_0 , x_0 \right) \in \mathcal{U}_{\bar{x}}, \forall \, t \in \left[ t_0 , \infty \right)$. This is the weakest stability condition we can expect for the dynamical system defined in Eq.~(\ref{dynsystem}).

Now, following \cite{Pugh,Jonsson}, the existence of a region such as $\mathcal{W}_{\bar{x}_{\mu}}$, or any similarly defined region due to a stronger stability condition, implies the existence of the region $\left\{f_{\mu}\right\}^{-1} \left( \mathcal{W}_{\bar{x}_{\mu}} \right) := \mathcal{N} \left( f_{\mu} \right) \subset \mathcal{M}$ consisting of \textit{Osedelec points} or \textit{regular points}. Such regions have well defined \textit{Lyapunov exponents} and represents an \textit{invariant Borel set of total measure} that possesses a collection of \textit{slowly varying Borel zero-forms} $R_{\varepsilon}: \mathcal{N} \left( f_{\mu} \right) \rightarrow \left( 1 , \infty \right), \varepsilon > 0$. Note that the argument can be reversed, due to the diffeomorphic character of the map $f_{\mu}$. Hence, the existence of the collection $R_{\varepsilon}$ of slowly varying Borel zero-forms, implies at least the existence of an associate dynamical system who's flow revolves around a uniform Lyapunov stable point.

%%%%%%%%%%%%%%
%%%%%%%%%%%%%%

%% References with bibTeX database:
%\section*{Bibliography}
\bibliographystyle{apsrev4-1}
\bibliography{sample.bib}
\end{document}